\documentclass[preprint,12pt]{amsart}

\usepackage{amsthm,amsfonts,amsmath,amssymb,amsrefs,bookmark,appendix}
\usepackage{amsxtra, mathrsfs}
\usepackage{colordvi}
\usepackage[usenames,dvipsnames]{color}
\usepackage{dsfont}
\usepackage{graphics}
\usepackage{tikz}

\theoremstyle{plain}
\newtheorem{theorem}{Theorem}
\newtheorem{lemma}[theorem]{Lemma}
\newtheorem{corollary}[theorem]{Corollary}
\newtheorem{proposition}[theorem]{Proposition}
\theoremstyle{remark}
\newtheorem{rem}[theorem]{\bf Remark}

\numberwithin{equation}{section}

\DeclareMathOperator{\spec}{spec}
\DeclareMathOperator{\ess}{ess}
\DeclareMathOperator{\supp}{supp}
\DeclareMathOperator{\loc}{loc}
\DeclareMathOperator{\dist}{dist}
\DeclareMathOperator{\dom}{dom}
\DeclareMathOperator{\vol}{vol}
\DeclareMathOperator{\sing}{sing}

\DeclareMathOperator{\coker}{coker}
\DeclareMathOperator{\st}{st}
\DeclareMathOperator{\link}{Link}
\DeclareMathOperator{\tr}{tr}

\DeclareMathOperator{\Hopf}{Hopf}

\renewcommand{\phi}{\varphi}
\newcommand{\eps}{\varepsilon}

\newcommand{\norm}[1]{\lVert #1 \rVert}

\newcommand{\cip}[2]{\langle #1, #2 \rangle}
\newcommand{\wt}[1]{\widetilde{#1}}
\newcommand{\nn}{\nonumber}

\newcommand{\D}{\mathbb{D}}
\newcommand{\R}{\mathbb{R}}
\newcommand{\C}{\mathbb{C}}
\newcommand{\Z}{\mathbb{Z}}
\newcommand{\T}{\mathbb{T}}
\newcommand{\N}{\mathbb{N}}
\renewcommand{\S}{\mathbb{S}}

\newcommand{\cA}{\mathcal{A}}

\newcommand{\cC}{\mathcal{C}}
\newcommand{\cD}{\mathcal{D}}
\newcommand{\cE}{\mathcal{E}}

\newcommand{\cG}{\mathcal{G}}
\newcommand{\cH}{\mathcal{H}}

\newcommand{\cL}{\mathcal{L}}
\newcommand{\cM}{\mathcal{M}}
\newcommand{\cN}{\mathcal{N}}

\newcommand{\rT}{\mathrm{T}}
\newcommand{\Tf}{[0,1]_{\mathrm{per}}}
\renewcommand{\d}{\mathrm{d}}

\newcommand{\sD}{\mathscr{D}}
\newcommand{\sF}{\mathscr{F}}
\newcommand{\sK}{\mathscr{K}}

\newcommand{\sS}{\mathscr{S}}

\newcommand{\Oms}{\Omega_{S}}
\newcommand{\ul}{\underline{\lambda}}
\newcommand{\ua}{\underline{\alpha}}
\newcommand{\uS}{\underline{S}}

\newcommand{\bA}{\boldsymbol{A}}
\newcommand{\bB}{\boldsymbol{B}}
\newcommand{\bG}{\boldsymbol{G}}
\newcommand{\bN}{\boldsymbol{N}}
\newcommand{\bT}{\boldsymbol{T}}

\newcommand{\bn}{\boldsymbol{n}}
\newcommand{\bx}{\boldsymbol{x}}
\newcommand{\be}{\boldsymbol{e}}
\newcommand{\bv}{\boldsymbol{v}}
\newcommand{\bp}{\boldsymbol{p}}

\newcommand{\bS}{\boldsymbol{S}}

\newcommand{\bsigma}{\boldsymbol{\sigma}}
\newcommand{\SO}{\boldsymbol{\mathrm{SO}}}
\newcommand{\SU}{\boldsymbol{\mathrm{SU}}}

\newcommand{\wtd}{\wt{\boldsymbol{\sigma}}(-i\wt{\nabla}) }

\begin{document}
\title[]{Self-adjointness and spectral properties of Dirac operators with magnetic links}

\date{}

\author[F. Portmann]{Fabian Portmann}
\address[F. Portmann]{QMATH, Department of Mathematical Sciences, University of Copenhagen
Universitetsparken 5, 2100 Copenhagen, DENMARK\newline
Current Address: IBM Switzerland, Vulkanstrasse 106 Postfach, 8010 Z\"{u}rich, Switzerland} 
\email{f.portmann@bluewin.ch}

\author[J. Sok]{J\'er\'emy Sok}
\address[J. Sok]{QMATH, Department of Mathematical Sciences, University of Copenhagen
Universitetsparken 5, 2100 Copenhagen, DENMARK\newline
Current Address: DMI, Universit\"{a}t Basel, Spiegelgasse 1, 4051, Basel, Switzerland} 
\email{jeremyvithya.sok@unibas.ch}

\author[J. P. Solovej]{Jan Philip Solovej}
\address[J. P. Solovej]{QMATH, Department of Mathematical Sciences, University of Copenhagen
Universitetsparken 5, 2100 Copenhagen, DENMARK} 
\email{solovej@math.ku.dk}

\thanks{The authors acknowledge support from the ERC grant Nr.\ 321029 ``The mathematics of the structure of matter" and
VILLUM FONDEN through the QMATH Centre of Excellence grant. nr.\ 10059.
This work has been done when all the authors were working at the university of Copenhagen.}

\begin{abstract}
	We define Dirac operators on $\S^3$ (and $\R^3$) with magnetic fields supported on smooth, oriented links
	and prove self-adjointness of certain (natural) extensions. We then analyze their spectral properties
	and show, among other things, that these operators have discrete spectrum. Certain examples, 
	such as circles in $\S^3$, are investigated in detail and we compute the 
	dimension of the zero-energy eigenspace.
\end{abstract}


\maketitle
\tableofcontents



\section{Introduction}
In this paper we study Dirac operators on $\S^3$ with magnetic fields that are 
supported on smooth, oriented links $\gamma \subset \S^3$.
These magnetic fields can be considered as generalizations
of the celebrated Aharonov-Bohm
magnetic solenoids allowing for more general interlinking curves and not just a line. 
Their singular nature (they are in some sense measures) 
makes the definition of the associated Dirac operator
a delicate task, which is fundamentally different from the smooth case.
Since the construction is local, it also works on $\R^3$
-- we have however chosen to work on $\S^3$ because then the Dirac operators
have discrete spectrum.

\subsection{Overview}
On $\R^3$, Dirac operators with singular \textit{electric potentials} have been studied before, see for
example \cite{DittEx89,DittEx92} for a study of Dirac operators with $\delta$-potentials
supported on shells, or \cite{Vegaetal12} for more general measure-valued electric 
potentials (which are singular with respect to the Lebesgue measure).
Furthermore, results for the Laplace operator with electric perturbations supported
on curves are investigated in e.g. \cite{ExYosh01, ExKond02, ExKond04, ExKond08}, 
as well as on surfaces in e.g. \cite{Brascheetal94, Behrndtetal13}.
In \cite{Baer00}, the author also investigates Dirac operators related to links in $\S^3$,
yet his setup is fundamentally different, as he studies Dirac operators in the complement 
of a link in $\S^3$ endowed with a hyperbolic metric.
Yet, to the best of our knowledge, there exist no results (neither on $\R^3$ nor on $\S^3$)
concerning Dirac operators with magnetic fields supported on links.

The situation is somewhat different in the two-dimensional case. 
After the discovery of the Aharonov-Bohm effect \cite{AhaBohm59} and the subsequent interest
in point-like magnetic fields, spectral properties 
of Dirac operators with $\delta$-type magnetic fields have
been investigated in numerous situations, see for 
example \cite{Gerbert89,Arai93,HiroOgur01,Tamura03, AraiHaya05},
as well as for the Pauli operator \cite{ErdVug02, Persson06}.
It should be noted that the spectral properties of Dirac operators on
even-dimensional manifolds are somewhat more accessible, since there are 
many more tools available.
In the smooth case, the Atiyah-Singer Index Theorem \cite{AS68} can provide useful information about the 
index of an operator. With the help of vanishing theorems one can reduce the computation of the index to a 
study of the zero-energy eigenspace. In some cases the latter can be analyzed efficiently through the 
Aharonov-Casher theorem, see \cite{AhaCash79} for the discussion on $\R^2$ 
and e.g. \cite{Cyconetal87, MR1860416} for the problem on $\S^2$.

Describing the zero-energy eigenspace on a three-dimensional manifold has so far proven to 
be a rather difficult task. A first step on $\R^3$ was undertaken 
in \cite{LossYau86} (see also subsequent works \cite{Adametal99, Adametal00, Elton00}), where the authors exhibited a 
square-integrable magnetic field for which the Dirac operator had non-trivial kernel. 
This discovery turned out to be important in the study of the stability of matter interacting with magnetic fields 
\cite{FroLiebLoss86, LiebLoss86}, or more precisely, that stability may fail depending on the 
physical parameters, e.g. the fine structure constant \cite{LiebLossSol95, Fefferman95, Fefferman96}.
The geometry behind the construction in \cite{LossYau86} was exploited in \cite{Adametal00_3} and in \cite{MR1860416}; 
the idea is to pull back smooth magnetic fields from $\S^2$ to $\S^3$ via Hopf maps
to obtain a non-trivial kernel for the Dirac operator on $\S^3$.
Through this mechanism, given any $m \in \N$, one
can find a smooth magnetic field on $\S^3$ such that the kernel of the associated Dirac operator has
exactly dimension $m$. Combining this result with the invariance of the kernel of the Dirac operator
under conformal maps -- here the stereographic projection -- yields the same answer for magnetic fields
on $\R^3$, and includes in particular the result from \cite{LossYau86}.

The idea from \cite{MR1860416} has somewhat served as a starting point for our analysis.
By extending the main result of \cite{MR1860416} to the case where the magnetic field on 
$\S^2$ contains a collection of point measures, we obtain a rather explicit
description of the spectrum of the Dirac operator on $\S^3$ with
a magnetic field supported on oriented interlinking circles, the Hopf links 
(see Theorem~\ref{thm:erdos-solovej}).
After this insight one is able to draft up a scheme to define the 
Dirac operator with a magnetic field
supported on an arbitrary oriented link in $\S^3$.

\subsection{Main results}
Consider a smooth, oriented knot $\gamma \subset \S^3$ with flux $2\pi\alpha \geq 0$.
To define the associated Dirac operator we choose a singular gauge.
For any oriented knot $\gamma$ there exists a smooth, compact and 
oriented surface in $S \subset \S^3$ such that $\partial S = \gamma$. 
The singular gauge is then defined as the Seifert surface weighted 
by the flux $2\pi\alpha$.
This gauge imposes a phase jump of $e^{-2i\pi\alpha}$ in the wave function 
across $S$, and the operator acts like the free Dirac operator
away from the surface. This essentially describes the minimal 
operator (see Section~\ref{sec:minimal_dirac_knot}), whose domain
can be characterized by
$$
	\{\psi \in H^1(\S^3 \setminus S)^2 : \left.\psi\right|_{S_+}=e^{-2i\pi\alpha}\left.\psi\right|_{S_-} \},
$$
where $S_{\pm}$ denote the top (resp. bottom) of the Seifert surface.
In particular, the Dirac operator exhibits a $2\pi$-periodicity
for the flux carried by $\gamma$. One readily observes that the 
phase jump alone is \emph{not} sufficient to obtain 
a self-adjoint operator, and we are lead to investigate different self-adjoint extensions
described by the behavior of the functions close to $\gamma$.
We are interested in two special extensions, namely where the singular part of the spinor
``aligns with or against the magnetic field". The first main result of this paper is
Theorem~\ref{thm:dirac_sa_knot}, where we prove the self-adjointness of these operators.
This result is then generalized to links in Theorem~\ref{thm:dirac_sa_link}.
In Section~\ref{sec:relation} we discuss the relation between these two self-adjoint extensions;
the first extension for $\gamma$ with flux $0<2\pi\alpha<2\pi$ is equal to the second
extension for the same knot $\gamma$ with \textit{opposite} orientation with flux $0<2\pi(1-\alpha)<2\pi$.
In \cite{dirac_s3_paper3} we will show that the extension where the spinor aligns against
the field is the physical extension, as it can be approximated by smooth magnetic fields
in the strong resolvent sense.

After having proved self-adjointness, we introduce an appropriate
topology on the space of Seifert surfaces and investigate the continuity properties of these
Dirac operators when varying the fluxes and the Seifert surfaces.
In Theorem~\ref{thm:compactness} we derive a compactness result and
in Theorem~\ref{thm:str_res_cont} we show that for a converging sequence of Seifert surfaces and fluxes,
the associated sequence of Dirac operators converges in the strong resolvent sense.
An immediate consequences of the compactness result is the discreteness of the spectrum
for the Dirac operator with an arbitrary magnetic link (Theorem~\ref{thm:discrete_spectrum}).
It also ensures that the kernel of the Dirac operator is trivial when the fluxes are small enough 
(Corollary~\ref{coro:no_zero_modes_for_small_fluxes}). 

At last, we study the kernel of our Dirac operators in two examples,
namely when the magnetic link is a circle or a Hopf link. 
In Corollary~\ref{cor:dim_ker_hopf_link}, following \cite{MR1860416}
we express the dimension of the kernel in the case of a Hopf link with specific fluxes.
In Theorem~\ref{thm:no_zero_modes_for_circles} we show that
the Dirac operator associated to a magnetic circle always has a trivial kernel. This fact will be
heavily used in two forthcoming publications \cite{dirac_s3_paper2, dirac_s3_paper3}.

Indeed, in this paper we are in some sense also laying the groundwork for two follow-up works.
The aforementioned $2\pi$-periodicity in the flux $2\pi\alpha$
suggests the study of the spectral flow of loops of Dirac operators, as $2\pi\alpha$ runs from $0$ to $2\pi$.
In \cite{dirac_s3_paper2} we will study the spectral flow for such paths of operators and obtain as 
a byproduct information about zero modes for the associated Dirac operators.
The study of these newly discovered zero modes is then
continued in \cite{dirac_s3_paper3}, where we provide a regularization procedure to show that 
these zero modes persist when the singular magnetic fields are approximated by smooth versions.
From elliptic regularity and the conformal invariance of the kernel of the Dirac operator we obtain
new examples of smooth, compactly supported magnetic fields on $\R^3$ for which the
associated Dirac operator has non-trivial kernel.

\subsubsection*{\bf{Notation:}}
Recall the Pauli matrices:
\begin{align*}
	\sigma_1 = \left(\begin{array}{cc}0 & 1 \\1 & 0\end{array}\right), 
	\quad \sigma_2 = \left(\begin{array}{cc}0 & -i \\i & 0\end{array}\right),
	\quad \sigma_3 = \left(\begin{array}{cc}1 & 0 \\0 & -1\end{array}\right).
\end{align*}
The standard inner product on $\C^2$ will always be denoted by $\cip{\cdot\,}{\cdot}$.

For any open set $\Omega$ in $\S^3$, 
the space of smooth functions on $\Omega$ is denoted by $\sD(\Omega)$.
Furthermore, we denote by $L^2(\Omega)^2$ the space of $L^2(\Omega)$-functions with values in $\C^2$.
The inner product on $L^2(\Omega)^2$ is given by
$$
	\cip{\psi}{\phi}_{L^2} := \int_{\Omega}\cip{\psi}{\phi},
$$
and the integration is performed with respect to the volume form of the underlying space.
For any $m \in \Z$, we denote by $\langle m \rangle := \sqrt{1+m^2}$ the Japanese bracket.

\section{Links and geometrical structure on $\S^3$}\label{sec:b_field_link}
We will always view the Riemannian manifold $\S^3$ as the unit sphere in $\C^2$ 
endowed with the standard metric $g_3$ from its ambient space.
The tangent bundle of $\S^3$ is denoted by $\rT\,\S^3$, and the class of smooth vector fields by $\Gamma(\rT\,\S^3)$. 
Since $\S^3$ can be seen as a Lie group, it admits a global orthonormal frame, 
which in turn defines an orientation. We choose to work in the orientation which
is mapped to the standard orientation of $\R^3$ under the stereographic projection 
(for reasons to become clear during the course of the exposition).

\subsubsection*{\bf{Convention:}}
For the remainder of this paper, by a link $\gamma \subset \S^3$ 
we always mean a \emph{smooth}, \emph{oriented} 
submanifold of $\S^3$ which is diffeomorphic to the \emph{disjoint} union of 
finitely many copies of $\S^1$
(and we henceforth refrain from explicitly stating that the link is smooth and oriented). 
Furthermore, we will often write
$\gamma = \bigcup_k \gamma_k$, 
where the $\gamma_k$ are knots (the connected components of $\gamma$).

\subsection{Links as magnetic fields on $\S^3$}\label{ssec:magn_links}
We begin by defining what we mean by a \textit{magnetic knot}.
Given a knot $\gamma \subset \S^3$, we would like to see it as a single field line with flux\footnote{In the case of negative flux, 
we simply study the knot with opposite orientation with flux $2\pi|\alpha|$.} $2\pi \alpha \geq 0$. 
Observe that a magnetic field
on $\S^3$ is normally a two-form, but due to the singular nature of our objects, it is more convenient to view the magnetic knot
$\bB$ as a one-current (see \cite[Chapter~3]{Rham55} for an introduction to currents), 
obtained through a combination of the Hodge dual and the musical isomorphisms 
$$
	\flat:
	\begin{array}{ccl}
	\rT\,\S^3 &\longrightarrow& \rT^{*}\S^3\\
	(\bp,X) &\mapsto& (\bp,g_3(X,\,\cdot))
	\end{array}
$$
and its inverse $\sharp$.
$\bB = 2\pi\alpha [\gamma]$ then acts on smooth one-forms $\omega \in \Omega^1(\S^3)$ as
$$
	(\bB;\omega) := 2\pi \alpha \int_{\gamma}\omega.
$$
From this definition it is easy to check the closed-condition; in the language of currents, the differential
of $\bB$ corresponds to the boundary operator, hence for $\phi \in \sD(\S^3)$ we have
\begin{align*}
	(\partial \bB;\phi) = (\bB;\d\phi)
	= 2\pi\alpha \int_{\gamma} \d \phi = 0.
\end{align*}

As in the classical electromagnetic setting, we would like to construct a magnetic gauge potential $\bA$ to 
$\bB$ such that $\partial \bA = \bB$ (in the language of currents). This can be done as follows;
by Seifert's theorem \cite{Seifert35, FP30} we know that given any (oriented) knot $\gamma \subset \S^3$, 
there exists a smooth, connected and oriented surface $S \subset \S^3$ such that 
$\partial S = \gamma$ (as oriented submanifolds).
On $\omega \in \Omega^2(\S^3)$, $\bA = 2\pi\alpha[S]$ acts as
\begin{align}\label{def:magn_gauge_def}
	(\bA;\omega) := 2\pi \alpha \int_{S} \omega.
\end{align}
For any $\omega \in \Omega^1(\S^3)$ we then have by Stokes theorem that
\begin{align}\label{eq:magn_gauge_pot}
	(\partial \bA;\omega) = (\bA;\d \omega) = 2\pi\alpha \int_{S}\d \omega
	= 2\pi\alpha \int_{\gamma} \omega = (\bB;\omega).
\end{align}
Note that if $\bA'$ is another magnetic gauge potential for $\bB$ (meaning $\partial S' = \gamma$ 
with the correct orientation and $\alpha' = \alpha$),
then they should be related by a boundary term. In the case when $S$ and $S'$ only intersect on $\gamma$, 
a quick computation shows that
\begin{equation}\label{eq:a_gauge_transf}
	\bA - \bA' = 2\pi\alpha\, [\partial V],
\end{equation}
where $[V]$ is the (signed) volume enclosed by the two surfaces
$S$ and $S'$, seen as a three-current. Should the surfaces intersect on more than $\gamma$, one can
use the result of \cite{MR916076} to construct a sequence of Seifert surfaces $S=S_0, S_1, \dots, S_n = S'$ such that
$S_i \cap S_{i-1} = \gamma$, $1 \leq i \leq n$. The gauge transformation is then performed stepwise from $S_{i-1}$
to $S_i$ using the previous result.

The generalization to a \textit{magnetic link} is now straight forward; given a link 
$\gamma = \bigcup_k \gamma_k \subset \S^3$, the magnetic link $\bB$ is defined as
$$
	\bB = \sum_{k}2\pi \alpha_k [\gamma_k], \quad \alpha_k \geq 0.
$$
A natural choice for the magnetic gauge potential is $\bA = \sum_{k}2\pi\alpha_k[S_k]$, 
where $S_k \subset \S^3$ is any Seifert surface with $\partial S_k = \gamma_k$. Note that if the fluxes on
two components of the link coincide, one may choose a common Seifert surface for them.

\subsection{The Seifert frame}\label{ssec:seif_frame}
Given any Seifert surface $S \subset \S^3$, there exists a natural construction
to endow the boundary $\partial S=\gamma$ with a natural moving frame, the
\textit{Seifert frame}. It is given by the triple 
$(\bT, \bS, \bN)$, where $\bT$ is the tangent vector to $\gamma$, $\bN$ is the normal of the oriented Seifert surface
and $\bS$ is the unit vector such that $(\bT, \bS, \bN)$ is positively oriented (thus along the link it points towards the
inside of the Seifert surface). 
On each component of $\gamma$, the Seifert frame coincides with the Darboux frame and
satisfies the Darboux equations
$$
	\dfrac{\mathrm{D}}{\d s}\begin{pmatrix} \bT\\ \bS\\ \bN\end{pmatrix}
	=\begin{pmatrix}0 & \kappa_g & \kappa_n \\ -\kappa_g & 0 & \tau_r\\ -\kappa_n & -\tau_r & 0 \end{pmatrix}
	\begin{pmatrix} \bT\\ \bS \\ \bN\end{pmatrix},
$$
where $\kappa_g, \kappa_n$ are the geodesic and normal curvature respectively, and $\tau_r$ 
is the relative torsion, see \cite[Chapter~7]{Spivakvol4} for more details.

For the coming chapters it will be necessary to extend the Seifert frame (which is apriori only defined on $\gamma$)
to a tubular neighborhood 
$$
	B_{\eps}[\gamma]:=\{\bp \in \S^3: \dist_{g_3}(\bp,\gamma)<\eps\}
$$
of the curve. The parameter $\eps>0$ is chosen such that the map
\begin{equation}\label{eq:exp_tubular_nb}
	\exp: 
	\begin{array}{ccl}
		B_{\eps}[\sigma_0;\cN\gamma] &\longrightarrow& B_{\eps}[\gamma]\\
		(\gamma(s),\bv(s)) &\mapsto& \exp_{\gamma(s)}(\bv(s))
	\end{array}
\end{equation}
is a diffeomorphism, where $\cN\gamma$ is the normal fiber to $\gamma$ in $\S^3$ and $\sigma_0$ its null-section.
We define the frame at a point $\bp = \exp(t_0\bv_0(s))\gamma(s)$, $\norm{\bv_0}=1$, 
by parallel transporting $(\bT,\bS,\bN)$ along the geodesic 
\begin{align}\label{eq:geodesic_s3}
	t \in [0,t_0] \mapsto \exp_{\gamma(s)}(t\bv_0(s)) = \cos(t)\gamma(s) + \sin(t)\bv_0(s),
\end{align}
which is a great circle in $\S^3$.
Furthermore, the geodesic distance $t_0$ will henceforth be denoted as $\rho_{\gamma}(\bp) = \rho(\bp)$.

\subsection{$Spin^c$ spinor bundles on $\S^3$}
We now have to define a Dirac operator with a singular magnetic potential $\bA=\sum_k 2\pi\alpha_k [S_k]$.
Such an operator is defined on the $L^2$-sections of a vector bundle, called a $Spin^c$ spinor bundle.
We review in this part its definition and its properties.

A $Spin^c$ spinor bundle $\Psi$ over a 3-manifold $\cM$ is a two-dimensional complex vector bundle over 
$\cM$ together with an inner product $\cip{\cdot\,}{\cdot}$ and an isometry $\bsigma: \rT^*\cM \to \Psi^{(2)}$, 
called the Clifford map, where
$$
	\Psi^{(2)} := \{M \in \mathrm{End}(\Psi): M = M^*, \tr(M)=0\}.
$$
A connection $\nabla^{\Psi}$ on $\Psi$ is a $Spin^c$ connection if for any vector field $X \in \Gamma(\rT\,\cM)$ it satisfies
\begin{enumerate}
	\item $X\cip{\xi}{\eta} = \cip{\nabla_X^{\Psi}\xi}{\eta} + \cip{\xi}{\nabla_X^{\Psi}\eta}$, 
	$\forall \xi,\eta \in \Gamma(\Psi)$.
	
	\item $[\nabla_X^{\Psi},\bsigma(\omega)] = \bsigma(\nabla_X \omega)$, $\forall \omega \in \Omega^1(\cM)$.
\end{enumerate}
Here, $\nabla_X$ is the Levi-Civita connection on $\S^3$.
It is well known that, up to isomorphism, the $Spin^c$ spinor bundle on $\S^3$ is unique and it is the trivial bundle
$\Psi = \S^3 \times \C^2$. For completeness we quickly sketch a proof of this fact.

Let $\Psi$ be a $Spin^c$ spinor bundle over $\S^3$ with Clifford map $\bsigma$. Since $\S^3$ is a Lie group,
it is parallelizable and we can choose a global orthonormal frame $\{\be_1,\be_2,\be_3\}$.
For each $\bp \in \S^3$, the fiber $\Psi_{\bp}$ can be decomposed into two complex lines
$$
	\Psi_{\bp} = \ker(\bsigma(\be^3(\bp))-1) \overset{\bot}{\oplus} \ker(\bsigma(\be^3(\bp))+1)
	=: \Lambda_{\bp}^{+}(\be^3) \overset{\bot}{\oplus}\Lambda_{\bp}^{-}(\be^3).
$$
This induces a decomposition of $\Psi$ into two complex line bundles 
$$
	\Psi = \Lambda^{+}(\be^3) \oplus \Lambda^{-}(\be^3).
$$
As the second homotopy group $\pi_2(\S^1)$ is trivial, any complex line bundle over $\S^3$ is trivial.
In particular, $\Lambda^{+}(\be^3)$ admits a global section
$\lambda_+$, which we can choose to have unit length at any point $\bp \in \S^3$. Similarly, we can define a global section
$\lambda_-$ on $\Lambda^{-}(\be^3)$, and fix their relative phase by requiring that
$$
	\omega(\be_1) + i\omega(\be_2) = \cip{\lambda_-}{\sigma(\omega)\lambda_+}, 
	\quad \forall \omega \in \Omega^1(\S^3).
$$
This then yields the desired trivialization. 

If $\nabla^{\Psi}$ is a $Spin^c$ connection, its connection form
is denoted by
\begin{align*}
	M_{\lambda}(X) := \begin{pmatrix} \cip{\lambda_+}{\nabla_X^{\Psi} \lambda_+} & \cip{\lambda_+}{\nabla_X^{\Psi} \lambda_-} \\
	\cip{\lambda_-}{\nabla_X^{\Psi} \lambda_+} & \cip{\lambda_-}{\nabla_X^{\Psi} \lambda_-}\end{pmatrix},
\end{align*}
where $X \in \Gamma(\rT\,B_{\eps}[\gamma])$. We then have 
(as in \cite[Prop.~2.9]{MR1860416})
\begin{multline}\label{eq:spin_basis}
	M_{\lambda}(X)\\
	= \frac{i}{2}\begin{pmatrix} \cip{\be_1}{\nabla_X \be_2} & -\cip{\be_3}{\nabla_X \be_2} - i\cip{\be_3}{\nabla_X \be_1} \\
	-\cip{\be_3}{\nabla_X \be_2} + i\cip{\be_3}{\nabla_X \be_1} & -\cip{\be_1}{\nabla_X \be_2}\end{pmatrix}\\
	- i\omega_{\lambda}(X)\mathrm{Id}_{\C^2},
\end{multline}
where $\omega_{\lambda}$ is the (local) real one form
$$
	\omega_{\lambda}(X) := \frac{i}{2}\left(\cip{\lambda_+}{\nabla_X\lambda_+} 
	+ \cip{\lambda_-}{\nabla_X\lambda_-}\right).
$$
By \cite[Prop.~2.11]{MR1860416}, the connection form of the induced connection from the Levi-Civita connection on 
$\S^3$ is given by \eqref{eq:spin_basis} with $\omega_{\lambda}=0$ (the connection form is then denoted by $M_{\lambda}^{\mathrm{can}}$).

\begin{rem}
	From now on we will use the convention that whenever $\nabla$ acts on a vector field, we mean the Levi-Civita connection,
	whereas if $\nabla$ acts on a spinor, then we mean the trivial (induced) connection on the $Spin^c$ bundle on $\S^3$.
	In particular, $-i\bsigma(\nabla)$ denotes the free Dirac operator on $\S^3$.
\end{rem}

\subsection{Charge conjugation}\label{charge_conj}
Observe that $\Psi$ is in fact a $Spin$ spinor bundle, meaning it is a $Spin^c$ spinor bundle together with an antilinear
bundle isometry $\mathrm{C}: \Psi \to \Psi$ such that $\cip{\eta}{\mathrm{C}\eta} = 0$ and $\mathrm{C}^2\eta = -\eta$ 
for all $\eta \in \Psi$. In the previous trivialization, the map $\mathrm{C}$ is given by
$$
	\mathrm{C}:
	\begin{array}{ccl}
		\S^3 \times \C^2 &\longrightarrow& \S^3 \times \C^2\\
		(\bp,\psi) &\mapsto & (\bp,i\sigma_2\overline{\psi}).
	\end{array}
$$
Furthermore, for any unital $\omega \in \Omega^1(\S^3)$ we can again define $\Lambda_{\pm}(\omega)$, and
one can show that $\mathrm{C}$ is a (fiber to fiber anti-linear) isomorphism between these two complex line bundles,
\begin{equation}\label{eq:C_interc}
	\mathrm{C}: \Lambda_{\pm}(\omega) \to \Lambda_{\mp}(\omega).
\end{equation}

The trivial connection $\nabla$ commutes with $\mathrm{C}$, hence it is a $Spin$ connection.
Similarly, one can show that the free Dirac operator $-i\bsigma(\nabla)$ also commutes with $\mathrm{C}$.

\section{The Dirac operator for a magnetic link}\label{sec:dirac_link}
For a smooth magnetic potential $\boldsymbol{\alpha}\in\Omega^1(\S^3)$, the associated $Spin^c$ connection on $\Psi$
in a trivialization  $(\lambda_+,\lambda_-)$ is given by the connection form
\footnote{The sign is different from \cite{MR1860416}, as our convention for the Dirac operator is different.} 
\cite[Proposition~2.11]{MR1860416}:
\[
	M_{\lambda}^{\mathrm{can}}(X)+i\boldsymbol{\alpha}(X)\mathrm{Id}_{\C^2}.
\] 
In this section, we extend the definition to the singular magnetic potential introduced in the previous section.

\subsection{Self-adjointness for a magnetic knot}\label{sec:sa_magn_knot}
\subsubsection{The minimal Dirac operator}\label{sec:minimal_dirac_knot}
Consider a magnetic knot $\bB$. To define the minimal Dirac operator, we choose the singular gauge 
$$
	\bA = 2\pi \alpha [S], \quad \alpha \geq 0,
$$
from \eqref{def:magn_gauge_def}.
Due to the singular nature of the magnetic gauge potential it is not possible to 
write the Dirac operator simply as $\bsigma(-i\nabla + \bA)$, and a more involved definition is necessary.

It is straight forward to see that away from the Seifert surface
on $\S^3 \setminus S$, the Dirac operator acts as the free
Dirac $-i\bsigma(\nabla)$. To determine the behavior close to the surface, we see 
(with the help of a gauge transformation \eqref{eq:a_gauge_transf}) 
that the elements of the domain need to have a phase jump of $e^{-2\pi i \alpha}$ across $S$. 

Denote by $S_+$ (resp. $S_-$) the top (resp. bottom) of the Seifert surface $S$ with
respect to $\bN$. To simplify the notation we define $\Oms:=\S^3\setminus S$.
We define the minimal Dirac operator via 
$$
	\cD_{\bA}^{(\min)}\psi:=-i \bsigma(\nabla)\psi_{\big|_{\Oms}}
$$
with domain
\begin{multline*}
	\dom\big(\cD_{\bA}^{(\min)}\big)
	:=\mathrm{clos}_{\cG}\big\{\psi\in \sD(\Oms)^2: 
	\supp\psi\subset\S^3\setminus\gamma,\\
	\left.\psi\right|_{S_\pm} \textrm{ exist \& are in }C^0(S)^2, \left.\psi\right|_{S_+}=e^{-2i\pi\alpha}\left.\psi\right|_{S_-}\big\},
\end{multline*}
where $\mathrm{clos}_{\cG}$ denotes the closure in the graph norm.
The derivative is taken in $[\sD'(\Oms)]^2$, but we have $\cD_{\bA}^{(\min)}\psi\in L^2(\Oms)^2$ 
and we canonically inject $\cD_{\bA}^{(\min)}\psi$ into $L^2(\S^3)^2$.

\begin{rem}
	Note that when given two fluxes $2\pi\alpha,2\pi\alpha' \geq 0$, the domains of the minimal Dirac operators
	agree if $2\pi(\alpha-\alpha') \in 2\pi \Z$. 
	We may therefore take $\alpha$ to be in 
	$\Tf:=\{[0,1]: 0 \sim 1\}$.
	A natural distance on $\Tf$ is given by
	\begin{equation}\label{def:dist_torus}
		\dist_{\Tf}(\alpha,\alpha') := \dist_{\R}(\alpha+\N,\alpha'+\N).
	\end{equation}
\end{rem}

Next, we define the domain $V_{S,\eps}:=\S^3\setminus\left\{\, S\cup B_{\eps}[\gamma] \right\}$.
\begin{proposition}[Generalized Stokes' formula]\label{prop:gen_stokes}
	Let $\{\be_1,\be_2,\be_3\}$ be an orthonormal frame for $\S^3$.
	Given $\psi,\phi\in H^1(V(S,\eps))^2$, we have
	\begin{align}\label{eq:general_stokes}
		&\int_{V_{S,\eps}}\cip{\bsigma(-i\nabla)\psi}{\phi} - \int_{V_{S,\eps}}\cip{\psi}{\bsigma(-i\nabla)\phi}\nn\\
		&\quad = \sum_{j=1}^3 \int_{\partial B_{\eps}[\gamma]} \iota_{\be_j} \left[\cip{i\bsigma(\be^j)\psi}{\phi} \vol_{g_3}\right]\\
		&\quad + \sum_{j=1}^3\int_{S \cap B_{\eps}[\gamma]^c} \iota_{\be_j} 
		\left[\Big(\cip{i\bsigma(\be^j)\left.\psi\right|_{S_+}}{\left.\phi\right|_{S_+}} 
		- \cip{i\bsigma(\be^j)\left.\psi\right|_{S_-}}{\left.\phi\right|_{S_-}}\Big) \vol_{g_3}\right].\nn
	\end{align}
	If $\psi\in \mathcal{A}_S$ where
	\begin{equation}\label{eq:def_A_S}
		\cA_{S} := \Big\{\psi \in [\sD'(\Omega_S)]^2 : -i\bsigma(\nabla)\psi \in L^2(\Omega_S)^2\Big\} \cap L^2(\Omega_S)^2,
	\end{equation}
	then the boundary terms $\psi|_{S_{\pm}}$ exist as elements in $H^{-1/2}_{\loc}(S\setminus\gamma)^2$.
	Furthermore \eqref{eq:general_stokes} still applies for $\phi\in H^1(\Omega_S)^2$ with $0<\eps<\dist(\supp\phi,\gamma)$,
	in which case the boundary term on $\partial B_\eps[\gamma]$ vanishes.
\end{proposition}
\begin{proof}
	Assume first $\psi,\phi\in H^1(\Omega_S)^2$.
	
	We write $-i\bsigma(\nabla) = -i\sum_{j=1}^{3}\bsigma(\be^j)\nabla_{\be_j}$ and
	set $f_j := \langle -i\bsigma(\be^j)\psi, \phi\rangle$.
	The Lie derivative along $\be_j$ is denoted by $\cL_{\be_j}$, and
	by Cartan's formula we have $\cL_{\be_j} = \d \circ \iota_{\be_j} + \iota_{\be_j} \circ \d$.

	Using the above representation of the Lie derivative in combination with Stokes' theorem yields
	$$
		\sum_{j=1}^3\int_{V_{S,\eps_+,\eps_-}}\cL_{\be_j}(f_j \vol_{g_3}) 
		= \sum_{j=1}^3\int_{\partial W_{S,\eps_+,\eps_-}} \iota_{\be_j} \left[\langle -i\bsigma(\be^j)\psi,\phi\rangle \vol_{g_3}\right],
	$$
	which after expansion gives the terms on the right-hand side of \eqref{eq:general_stokes}. On the other hand,
	\begin{align*}
		\sum_{j=1}^3\cL_{\be_j}(f_j \vol_{g_3}) &= \sum_{j=1}^3 \d f_j(\be_j)\vol_{g_3} + f_j \mathrm{div}_{g_3}(\be_j)\vol_{g_3}.
	\end{align*}
	Metric compatibility gives
	$$
		\sum_{j=1}^3\d f_j(\be_j) = \cip{-i\bsigma(\nabla)\psi}{\phi} - \cip{\psi}{-i\bsigma(\nabla)\phi} 
		+ \sum_{j=1}^3\cip{-i\bsigma(\nabla_{\be_j}\be^j) \psi}{\phi}.
	$$
	However, $\cip{-i\sum_{j=1}^3\left[\bsigma(\nabla_{\be_j}\be^j) + \mathrm{div}_{g_3}(\be_j)\bsigma(\be^j)\right]\psi}{\phi} = 0$,
	and the desired formula follows.
	
	Now assume $\psi \in \mathcal{A}_S$ and $\phi\in H^1(\Omega_S)^2$ with $0<\eps<\dist(\supp\phi,\gamma)$.
	For $0<\eps_+,\eps_-<\eps$ let $W_{S,\eps_1,\eps_2}\subset V_{S,\eps}$ be the subset of $V_{S,\eps}$
	obtained by removing the points $\bp$ that are either:
	\begin{enumerate}
	 \item above $S$ with $\dist(\bp,S)\le \eps_+$,
	 \item or below $S$ with $\dist(\bp,S)\le \eps_-$.
	\end{enumerate}
	
	As above, we use Stokes' formula, but on $W_{S,\eps_+,\eps_-}$ instead. The boundary term on $\partial B_\eps[\gamma]$ vanishes
	leaving only two corresponding to the boundary above and below $S$.
	
	By taking the limits $\eps_+\to 0$ with $\eps_-$ fixed,
	the boundary terms corresponding to $S_+$ has a well-defined limit.
	For any $\eps>0$, this is true for any $\phi\in H^1(\Omega_S)^2$ with $\dist(\supp\phi,\gamma)>\eps$.
	This shows that $\left.\psi\right|_{S_+}\in H^{-1/2}_{\loc}(S\setminus\gamma)$.
	By symmetry the same holds for $\left.\psi\right|_{S_-}$. We end the proof by taking the limit $\eps_+,\eps_-\to 0$.
\end{proof}

\begin{lemma}\label{lem:dmin_sym}
	The operator $\cD_{\bA}^{(\min)}$ is symmetric on its domain.
\end{lemma}
\begin{proof}
	Given any two $\psi,\phi \in \dom\big(\cD_{\bA}^{(\min)}\big)$, we have
	\begin{align*}
		\cip{\cD_{\bA}^{(\min)}\psi}{\phi}_{L^2}
		= \lim_{\eps \to 0} \int_{V_{S,\eps}} \cip{-i\bsigma(\nabla)\psi}{\phi}.
	\end{align*}
	We now use \eqref{eq:general_stokes}; the last boundary term vanishes since the functions in 
	$\dom\big(\cD_{\bA}^{(\min)}\big)$ only differ by a phase on $S$
	and the boundaries have opposite orientation. We therefore have
	\begin{align*}
		&\cip{\cD_{\bA}^{(\min)}\psi}{\phi}_{L^2}
		= \lim_{\eps \to 0} \int_{V_{S,\eps}} \langle -i\bsigma(\nabla)\psi,\phi\rangle\\
		&\quad = \lim_{\eps \to 0} \Big( \int_{V_{S,\eps}} \langle \psi,-i\bsigma(\nabla)\phi\rangle
		+ \int_{\{\rho_{\gamma}=\eps\}}\sum_{j=1}^3 \iota_{\be_j} \left[\langle -i\bsigma(\be^j)\psi, \phi\rangle \vol_{g_3}\right] \Big)\\
		&\quad = \cip{\psi}{\cD_{\bA}^{(\min)}\phi}_{L^2},
	\end{align*}
	where the remaining boundary term vanishes when both $\psi, \phi$ are supported away from the curve. 
	The general case follows by density.
\end{proof}

The Lichnerowicz formula 
\begin{equation}\label{eq:lichnerowicz}
	[\bsigma(-i\nabla)]^2=-\Delta_{\S^3}+\frac{1}{4}R_{\mathrm{sc}}^{\S^3}=-\Delta_{\S^3}+\frac{3}{2}
\end{equation}
suggests a relation between the minimal domain and the Sobolev functions with the phase jump across $S$.
Indeed, we have the following proposition, whose proof can be found in the
Appendix.
\begin{proposition}\label{prop:egalite_min_dom}
	Let $H^1_{\bA}(\S^3)\subset L^2(\S^3)$ be defined as
	\begin{multline*}
		H^1_{\bA}(\S^3):=
		\{f \in L^2(\S^3): (\nabla_Xf)_{|_{\Oms}} \in L^2(\Omega_S),\forall X \in \Gamma(\rT\,\S^3);\\
		\left.\psi \right|_{S_+}=e^{-2i\pi\alpha}\left. f \right|_{S_-}\in H^{1/2}(S)\},
	\end{multline*}
	which corresponds to the elements of $H^1(\Omega_S)$ with the phase-jump-condition across $S$.
	Then $H^1_{\bA}(\S^3)^2$ coincides with $\dom\big(\cD_{\bA}^{(\min)}\big)$, and 
	for any $\psi\in \dom\big(\cD_{\bA}^{(\min)}\big)$ we have the identity
	\begin{equation}\label{eq:dirac_gradient}
		\int |\cD_{\bA}^{(\min)}\psi|^2=\int |(\nabla\psi)_{|_{\Oms}}|^2+\frac{3}{2}\int |\psi|^2.
	\end{equation}
\end{proposition}

One readily observes that for $\alpha = 0$ there is no phase jump across the surface $S$, 
which implies that the closure of the minimal Dirac operator is actually the free Dirac on $H^1_0(\S^3\setminus\gamma)^2$.
Since $\gamma$ has co-dimension $2$ in $\S^3$, we have $H^1_0(\S^3\setminus\gamma)=H^1(\S^3)$ by a slight
modification of the arguments given in e.g. \cite{Sved81}, and we can conclude the following.
\begin{proposition}\label{prop:dmin_alpha_0}
	For $\alpha = 0$, the minimal Dirac operator $\cD_{\bA}^{(\min)}$ is self-adjoint 
	and it coincides with the free Dirac operator $-i\bsigma(\nabla)$ with domain $H^1(\S^3)^2$.
\end{proposition}

\subsubsection{Extensions for $0<\alpha<1$}
In order to characterize the self-adjoint extensions of the operator $\cD_{\bA}^{(\min)}$ 
in the singular gauge $\bA$, we need to specify the behavior of the functions 
(in the prospective domain) close to the knot $\gamma$. 
Up to fixing a base point $\bp_0$, we identify $\gamma$ with its arclength parametrization
$\gamma: \R/\ell\Z \to \S^3$ with $\gamma(0) = \bp_0$.

We are interested in characterizing two particular extensions, namely where the additional 
elements of the domain ``align with or against the magnetic field". 
To make this more rigorous, we first have to define two smooth spinors $\xi_{\pm}$, which correspond
to the normalized eigenfunctions of $\bsigma(\bT^{\flat})$:
\begin{align*}
	\bsigma(\bT^{\flat})\xi_{+} = \xi_{+}, \quad \bsigma(\bT^{\flat})\xi_{-} = -\xi_{-}.
\end{align*}
These spinors are a priori only defined on $\gamma$, but we can extend them in a tubular neighborhood of $\gamma$ by
parallel transport of the Seifert frame (hence of $\bT$) along the geodesics \eqref{eq:geodesic_s3}. Note that this defines a \textit{smooth} extension of $\xi_{\pm}$,
as parallel transport is a first order equation which depends smoothly on the initial conditions and the parameters of the equation.
In particular, $\nabla_{X}\xi_{\pm} \in C(B_{\eps}[\gamma];\Psi)$ $\forall X \in \Gamma(\rT\,\S^3)$.
Furthermore, we fix their relative phase by requiring that 
\begin{align}\label{eq:rel_phase}
	\omega(\bS) + i \omega(\bN) = \langle \xi_-, \bsigma(\omega)\xi_+\rangle, \quad \forall \omega \in \Omega^1(B_{\eps}[\gamma]).
\end{align}

Let $\chi: \R \to \R_+$ be a smooth function with $\supp \chi \in [-1,1]$ and $\chi(x) = 1$ for $x \in [-2^{-1},2^{-1}]$. 
Furthermore, we define the localization function at level $0<\delta \ll 1$ by
\begin{equation}\label{eq:chi_loc_curve}
	\chi_{\delta,\gamma}: 
	\begin{array}{ccl}
		B_{\delta}[\gamma] &\longrightarrow& \R_+\\
		\bp = \exp_{\gamma(s)}(\rho_{\gamma}\bv_0(s))
		&\mapsto& \chi\left(\rho_{\gamma}\delta^{-1}\right).
	\end{array}
\end{equation}
We then pick $\delta>0$ such 
that\footnote{This choice is motivated by the fact that for any $\bp \in \supp \chi_{\delta,\gamma}$ we have 
$\cos(\rho_{\gamma}) - \sin(\rho_{\gamma})|\kappa(s)| \geq 1-\eps$.}
$$
	0<\delta< \eps \cdot \min\left\{1, \left(\norm{\kappa}_{L^{\infty}} + \sqrt{\eps + \norm{\kappa}_{L^{\infty}}^2}\right)^{-1}\right\},
$$
where $\eps>0$ is such that \eqref{eq:exp_tubular_nb} is a diffeomorphism. Here $\norm{\kappa}_{L^{\infty}}$ denotes:
\[
	\sup_{\theta,s}\big| \cos(\theta)\kappa_g(s)+\sin(\theta)\kappa_n(s)\big|.
\]

\begin{theorem}\label{thm:dirac_sa_knot}
	Let $\gamma \subset \S^3$ be a knot with Seifert surface $S$, and set
	$\bA = 2\pi \alpha [S]$, $0<\alpha<1$.
	We define the Dirac operators $\cD_{\bA}^{(\pm)}$ by
	$$
		\left\{
		\begin{array}{lcl}
		\dom\big(\cD_{\bA}^{(\pm)}\big)
		&:=& \{\psi \in \dom\big((\cD_{\bA}^{(\min)})^*\big): 
		\langle\xi_{\mp}, \chi_{\delta,\gamma}\psi\rangle \xi_{\mp} \in \dom\big(\cD_{\bA}^{(\min)}\big)\},\\
		\cD_{\bA}^{(\pm)}\psi &:=& -i\bsigma(\nabla)\psi\big|_{\Oms} \in  L^2(\Omega_{S})^2 \hookrightarrow L^2(\S^3)^2.
		\end{array}
		\right.
	$$
	These definitions are independent of the choice of $\chi_{\delta,\gamma}$ and both operators 
	are self-adjoint.
\end{theorem}

\subsubsection{Relation between the two extensions $\cD_{\bA}^{(\pm)}$}\label{sec:relation}
For a magnetic vector potential $\bA = 2\pi \alpha [S]$, we define
$$
	-\bA := 2\pi\alpha[-S],
$$
and the charge-conjugate potential $\bA_{\mathrm{C}}$ as
$$
	\bA_{\mathrm{C}} := 2\pi (1-\alpha)[-S],
$$
where $-S$ denotes the surface $S$ with opposite orientation.

\begin{proposition}\label{prop:c_conj}
	The charge conjugation $\mathrm{C}$ maps $\dom(\cD_{\bA}^{(-)})$ 
	onto $\dom(\cD_{-\bA}^{(-)})$ and we have
	$$
		\mathrm{C}\, \cD_{\bA}^{(-)} = \cD_{-\bA}^{(-)}\,\mathrm{C}.
	$$
	Furthermore,
	$$
		\cD_{\bA}^{(+)}=\cD_{\bA_{\mathrm{C}}}^{(-)}.
	$$
\end{proposition}
\begin{proof}
	To show the first statement, we recall \eqref{eq:C_interc}: 
	the charge conjugation $\mathrm{C}$ interchanges pointwise the two complex lines
	$\C\xi_{\pm}(\bp)$. Furthermore, $\mathrm{C}$ commutes with the free Dirac operator $-i\bsigma(\nabla)$ 
	and both operators are \emph{local}. This then implies that $\cD_{\bA}^{(-)}$ and 
	$\cD_{-\bA}^{(-)}$ are anti-unitarily equivalent through $\mathrm{C}$.
	
	For the second claim we note that the phase jumps for both operators are the same 
	and the complex lines $\C\xi_{\pm}(\bp)$ 
	for $S$ coincide with the complex lines $\C \xi_{\mp}(\bp)$ for $-S$.
\end{proof}

\subsection{Proof of Theorem~\ref{thm:dirac_sa_knot}}
Before we begin with the proof, it is necessary to introduce some technical tools.

\subsubsection{Coordinates for $B_{\eps}[\gamma]$}\label{sec:coord}
We will need coordinates for the tubular neighborhood of the curve.
We have already obtained two coordinates, namely the geodesic distance from the
curve $\rho_{\gamma} =: \rho$ and the arc length parameter $s$. 
To efficiently describe the phase jump of the function, it is convenient to 
to construct an angle function $\theta$ on 
$B_{\eps}[\gamma] \setminus \gamma$. We define it with the help of the Seifert frame through
\begin{align}\label{eq:drho}
	\d \rho = \cos(\theta) \bS^{\flat} + \sin(\theta) \bN^{\flat}.
\end{align}
In other words, for $\bp \in B_{\eps}[\gamma]\setminus \gamma$ we have
$$
	\bp = \exp_{\gamma(s)}(\rho(\bp)[\cos(\theta)\bS(s) + \sin(\theta)\bN(s)]).
$$
As the Seifert frame gives a trivialization of the normal fibre of the curve, these coordinates provide the inverse
of the exponential map from the neighborhood of the null-section in this trivialization onto $B_{\eps}[\gamma]$.

From the definition it is clear that the push-forward $\exp_{*}(\partial_{\rho})$ is orthogonal to 
$\exp_{*}(\partial_s)$ and $\exp_{*}(\partial_\theta)$. From \eqref{eq:geodesic_s3} we also obtain
\begin{equation}\label{eq:partiald_coord}
	\left\{
	\begin{array}{rcl}
		\exp_{*}(\partial_{s}) &=& (\cos(\rho)-\sin(\rho))(\kappa_g(s) \cos(\theta) + \kappa_n(s)\sin(\theta))\bT(s)\\
		&& +\sin(\rho)\tau_r(s)\left[-\sin(\theta)\bS(s) + \cos(\theta)\bN(s)\right]\\
		\exp_{*}(\partial_{\rho}) &=& \cos(\theta)\bS(s) + \sin(\theta)\bN(s)\\
		\exp_{*}(\partial_{\theta}) &=& \sin(\rho)\left[-\sin(\theta)\bS(s) + \cos(\theta)\bN(s)\right].
	\end{array}\right.
\end{equation}
Setting $h(\bp) := \cos(\rho) - \sin(\rho)(\kappa_g(s)\cos(\theta) +\kappa_n(s)\sin(\theta))$, we get that
\begin{align}\label{def:bg}
	(\bT,(\d\rho)^{\sharp},\bG) :&=(\bT,\exp_{*}(\partial_{\rho}), \sin(\rho)^{-1}\exp_{*}(\partial_{\theta}))\\
	&= (h^{-1}(\exp_{*}(\partial_{s}) - \tau_r \exp_{*}(\partial_{\theta})),\exp_{*}(\partial_{\rho}), \sin(\rho)^{-1}\exp_{*}(\partial_{\theta}))\nonumber
\end{align}
is an orthonormal basis of $\rT\, (B_{\eps}[\gamma] \setminus \gamma)$.
The pullback of the volume form is given by
\begin{equation}\label{eq:pullb_volform}
	\exp^{*}(\vol_{g_3}) = \exp^{*}(\bT \wedge \d\rho \wedge \bG^{\flat}) = h \sin(\rho)\,\d s \wedge \d\rho \wedge \d\theta.
\end{equation}

\subsubsection{The model case}\label{ssec:T_straight_line}
Consider the manifold
$$
	\T_{\ell} \times \R^2 := (\R / (\ell \Z)) \times \R^2,
$$
equipped with the flat metric
and coordinates $(s,u_1,u_2)$. On this manifold we want analyze the Dirac operator with
magnetic field supported on the curve $\T_{\ell} \times\{0\}\subset \T_{\ell} \times\R^2$.
Understanding this particular model case will be of great use 
in the analysis of the Dirac operator whose magnetic field is supported on a general knot.
The idea is that for a general knot in $\S^3$, the different self-adjoint extensions will be characterized by the 
behavior of the spinors close to the curve. In a sufficiently small tubular neighborhood of the curve
we will see that the action of the (general) Dirac operator
is closely related to the model operator on $\T_{\ell} \times\R^2$.

The $Spin^c$ spinor bundle is $\T_{\ell} \times \R^2 \times \C^2$, endowed with the Clifford map $\wt{\bsigma}$ 
defined through
$$
	\wt{\bsigma}(\d s) = \sigma_3, \quad \wt{\bsigma}(\d u_1) = \sigma_1, \quad \wt{\bsigma}(\d u_2) = \sigma_2.
$$
Furthermore, the symbol $\wt{\nabla}$ denotes the canonical connection on the above $Spin^c$ spinor bunde.
We will study the problem in the (radial) gauge
\begin{equation}\label{def:std_gauge}
	\bA(s,u_1,u_2)
	:=\alpha \frac{-u_2\d u_1 + u_1\d u_2}{u_1^2+u_2^2}
	\in \rT^*_{(s,u)}\left(\T_{\ell}\times\R^2 \right),
\end{equation}
with $\alpha \in [0,1)$.
\begin{rem}\label{rem:dmod_singg}
	Note that $\bA = \alpha \d \theta$, where $e^{i\theta} = (u_1+iu_2)(u_1^2+u_2^2)^{-1/2}$. 
	It is linked to the singular gauge
	$$
		\bA_0[\omega] = 2\pi \alpha\int_{\{\theta = 0\}}\omega
	$$
	through the gauge transformation $e^{i\alpha \theta}$, where $\theta$ has branch cut along $\theta=0$.
	One can therefore write formally
	\begin{align}\label{eq:dmodel_gtraf}
		\wt{\bsigma}(-i\wt{\nabla} + \bA) = e^{-i\alpha\theta}\wt{\bsigma}(-i\wt{\nabla})e^{i\alpha\theta}.
	\end{align}
\end{rem}

The minimal operator $\cD_{\T_{\ell},\alpha}^{(\min)} := \wt{\bsigma}(-i\wt{\nabla} + \bA)$ has domain 
$$
	\dom(\cD_{\T_{\ell},\alpha}^{(\min)}) 
	:= \mathrm{clos}_{\mathcal{G}}\left(C_0^{\infty}(\T_{\ell} \times (\R^2 \setminus \{0\}))^2\right).
$$
According to von Neumann's theorem, all self-adjoint extensions can be obtained through 
the study of the deficiency spaces of the maximal operator
$$
	\cD_{\T_{\ell},\alpha}^{(\max)} := \big(\cD_{\T_{\ell},\alpha}^{(\min)}\big)^*.
$$
Define $\hat{H}^{-1}(\T_{\ell}) := \ell_2(\T_{\ell}^*;\langle j\rangle^{-2}\d\mu_{\T_{\ell}^*}(j))$, where
$\T_{\ell}^* := \tfrac{2\pi}{\ell}\Z$ is the Pontryagin dual of $\T_{\ell}$ and $\d\mu_{\T_{\ell}^*}$ is its counting measure.
\begin{lemma}\label{lem:model_c_def_sp}
	The deficiency spaces of $\cD_{\T_{\ell},\alpha}^{(\max)}$ are given by
	\begin{align*}
		&\ker\big(\cD_{\T_{\ell},\alpha}^{(\max)} \mp i\big)\\
		&\quad = \left\{f_{\pm i}(\underline{c}) := \frac{1}{\sqrt{2\pi\ell}}\sum_{j \in \T_{\ell}^*} c_j e^{ijs}
		\begin{pmatrix}K_{1-\alpha}(r\langle j\rangle)e^{-i\theta}\\
		\frac{(\pm1-ij)}{\langle j\rangle}K_{\alpha}(r\langle j\rangle)\end{pmatrix} 
		: \underline{c} \in \hat{H}^{-1}(\T_{\ell})\right\},
	\end{align*}
	where $K_{\alpha}$ and $K_{1-\alpha}$ denote the modified Bessel functions of the second kind.
\end{lemma}
\noindent The proof of this lemma is given in the Appendix. 
The domain of $\cD_{\T_{\ell},\alpha}^{(\max)}$ can therefore be written as
\begin{multline}\label{eq:dtmax_dom}
	\dom\big(\cD_{\T_{\ell},\alpha}^{(\max)}) = \dom(\cD_{\T_{\ell},\alpha}^{(\min)}) \overset{\perp_{\cG}}{\oplus}
	\Bigg\{\frac{1}{\sqrt{2\pi\ell}}
	\sum_{j \in \T_{\ell}^*}a_je^{ijs}\begin{pmatrix} K_{1-\alpha}(r\langle j\rangle)e^{-i\theta}\\
	-ij\langle j\rangle^{-1} K_{\alpha}(r\langle j\rangle)\end{pmatrix}\\
	+ \frac{b_j}{\langle j\rangle}e^{ijs}\begin{pmatrix} 0\\
	K_{\alpha}(r\langle j\rangle)\end{pmatrix}
	: \underline{a},\underline{b} \in \hat{H}^{-1}(\T_{\ell})
	\Bigg\},
\end{multline}
and all self-adjoint extensions of $\cD_{\T_{\ell},\alpha}^{(\min)}$ are characterized by isometries 
$$
	U: \ker(\cD_{\T_{\ell},\alpha}^{(\max)}-i) \to \ker(\cD_{\T_{\ell},\alpha}^{(\max)}+i),
$$
which in this case correspond to isometries on $\hat{H}^{-1}(\T_{\ell})$.

We are interested in two particular extensions $\cD_{\T_{\ell},\alpha}^{(\pm)}$ with domains
\begin{align*}
	&\dom\big(\cD_{\T_{\ell},\alpha}^{(+)} \big):=\dom\big(\cD_{\T_{\ell},\alpha}^{(\min)} \big)
	\overset{\perp_{\cG}}{\oplus}
	\left\{ \frac{1}{\sqrt{2\pi\ell}} \sum_{j \in \T_{\ell}^*} \lambda_j e^{ijs}
	\begin{pmatrix}K_{1-\alpha}(r\langle j\rangle)e^{-i\theta}\\0\end{pmatrix}: \ul \in\ell_2(\T_{\ell}^*)\right\}\\
	&\dom\big(\cD_{\T_{\ell},\alpha}^{(-)} \big):=\dom\big(\cD_{\T_{\ell},\alpha}^{(\min)} \big)
	\overset{\perp_{\cG}}{\oplus}
	\left\{ \frac{1}{\sqrt{2\pi\ell}} \sum_{j \in \T_{\ell}^*} \lambda_j e^{ijs}
	\begin{pmatrix}0\\K_{\alpha}(r\langle j\rangle)\end{pmatrix}: \ul \in\ell_2(\T_{\ell}^*)\right\},
\end{align*}
which correspond to the isometries
\begin{align*}
	U^{(+)}: f_{i}(\underline{c}) \mapsto f_{-i}\big(\big(\tfrac{1-ij}{1+ij}c_j\big)_j\big),
	\quad U^{(-)}: f_{i}(\underline{c}) \mapsto f_{-i}(-\underline{c}).
\end{align*}
These extensions should be thought of as the situation when the singular part of the 
spinor ``aligns" with or against the magnetic field.
For any of these extensions we can decompose $f \in \dom\big(\cD_{\T_{\ell},\alpha}^{(\pm)}\big)$ 
with respect to the splittings above as 
\begin{align}\label{eq:model_decomp}
	f = f_0+f_{\sing}(\ul)
\end{align}
with $f_0 \in \dom\big(\cD_{\T_{\ell},\alpha}^{(\min)} \big)$.
Furthermore, any $f \in \dom\big(\cD_{\T_{\ell},\alpha}^{(\pm)}\big)$ satisfies
\begin{align}\label{eq:dpm_decouple}
	\int |\cD_{\T_{\ell},\alpha}^{(\pm)}f|^2
	= \int|(\partial_s f)\big|_{\theta\neq 0}|^2 + \int |(\wt{\bsigma}(-i\wt{\nabla}_u + \bA_u)f)\big|_{\theta\neq 0}|^2,
\end{align}
which can be verified with the help of an explicit calculation 
(By Stokes' formula the equality is true for $f=f_0+f_{\sing}(\ul)$, where $f_0\in C^\infty_0(\T_\ell\times (\R^2\setminus\{0\}))$ and $\ul$
is a sequence which $\lambda_j=0$ for $|j|$ big enough, hence it is true for $f\in \dom\big(\cD_{\T_{\ell},\alpha}^{(\pm)}\big)$ 
by density in the graph norm).
Another computation yields
\begin{align*}
	&\cD_{\T_{\ell},\alpha}^{(+)}\Big(\sum_{j\in\T_\ell^*}\lambda_j e^{ijs}
	\begin{pmatrix}K_{1-\alpha}(r\langle j\rangle)e^{-i\theta}\\ 0 \end{pmatrix}\Big)
	=i\sum_{j\in\T_\ell^*} \lambda_j e^{ijs}\begin{pmatrix}jK_{1-\alpha}(r\langle j\rangle)e^{-i\theta}\\ 
	\langle j\rangle K_\alpha(r\langle j\rangle)\end{pmatrix},\\
	&\cD_{\T_{\ell},\alpha}^{(-)}\Big(\sum_{j\in\T_\ell^*}\lambda_j e^{ijs}
	\begin{pmatrix}0\\ K_{\alpha}(r\langle j\rangle) \end{pmatrix}\Big)
	=i\sum_{j\in\T_\ell^*} \lambda_j e^{ijs}\begin{pmatrix}\langle j\rangle K_{1-\alpha}(r\langle j\rangle)e^{-i\theta}\\ 
	-j K_\alpha(r\langle j\rangle) \end{pmatrix},
\end{align*}
so that we have a lower bound on the graph norm $\norm{\,\cdot\,}_{\T_{\ell}}$
\begin{equation}\label{eq:sing_contr}
	\norm{f}_{\T_{\ell}}^2=\norm{f_0}_{\T_{\ell}}^2+\norm{f_{\sing}}_{\T_{\ell}}^2
	\ge (C_\alpha+C_{1-\alpha}) \norm{\ul}_{\ell_2(d\mu_{\T_{\ell}^*})}^2.
\end{equation}
In other words, the graph norm  controls the $\ell_2$-norm of $\ul$. Remark also that we have:
\begin{equation}\label{eq:sing_psi_la}
	\int|f_{\sing}(\ul)|^2=\Big(\int_{0}^{+\infty}|K_\alpha(r)|^2rdr\Big)\sum_{j\in\T_\ell^*} \frac{|\lambda_j|^2}{1+j^2}.
\end{equation}

We now state several results that will come in handy when proving the self-adjointness of the operator 
in the general case.
Given $f, g \in \dom\big(\cD_{\T_{\ell},\alpha}^{(\max)}\big)$, we set
$$
	f =(f_+,\,f_-)^{\rT} , \quad g = (g_+,\,g_-)^{\rT}.
$$
Using Stokes' theorem, it is then easy to show that
\begin{multline}\label{eq:dtmax_ibp}
	\cip{\cD_{\T_{\ell},\alpha}^{(\max)}f}{g}_{L^2} - \cip{f}{\cD_{\T_{\ell},\alpha}^{(\max)}g}_{L^2}\\
	= -i \lim_{\rho \to 0}\int_0^{\ell}\int_{0}^{2\pi}\rho
	\left[e^{i\theta}\overline{f_{-}}g_{+} + e^{-i\theta}\overline{f_{+}}g_{-}\right](s,\rho,\theta)\, \d s\d\theta,
\end{multline}
where we have introduced polar coordinates $(\rho,\theta)$ for $(u_1,u_2) \in \R^2$.
Splitting $f=f_0 + f_{\sing}$ and $g = g_0 + g_{\sing}$ 
according to \eqref{eq:dtmax_dom}, we can replace $f, g$ in \eqref{eq:dtmax_ibp} by
their singular parts by using the defining properties of the maximal operator.
Note that
$$
	f \in \dom\big(\cD_{\T_{\ell},\alpha}^{(\pm)}\big) \quad\textrm{ if and only if }\quad (f_{\sing})_{\pm} \equiv 0,
$$
so that the boundary terms vanish when considering the respective extensions.

\subsubsection{Proof of Theorem~\ref{thm:dirac_sa_knot}}
	To keep the notation short, we denote 
	$$
		\cD_{\min} := \cD_{\bA}^{(\min)}, \quad \cD_{\max} := \big(\cD_{\bA}^{(\min)}\big)^*.
	$$
	The independence of the domain (resp. the operator) on the cutoff function is clear, since difference of two such cutoff functions
	$\chi_{\eps_0,\gamma}$ and $\wt{\chi}_{\eps_0',\gamma}$ will vanish on a small enough neighborhood of the curve.
	
	\medskip
	\noindent\textit{Analysis of $\dom(\cD_{\max})$:} 
	The domain of $\cD_{\max}$ is defined as follows: $\psi \in \dom(\cD_{\max})$ if there exists an $w \in L^2(\S^3)^2$ such that
	for any $\phi \in \dom(\cD_{\min})$ we have
	$$
		\cip{\psi}{\cD_{\min}\phi}_{L^2} = \cip{w}{\phi}_{L^2}.
	$$
	Restricted to $\Omega_S$, we get that $\psi \in \mathcal{A}_S$, see \eqref{eq:def_A_S}. We have the following characterization.
	
	\begin{lemma}\label{lem:charac_D_max}
	   For any pair $\psi\in \dom(\cD_{\max})$ and $u\in \sD(\S^3\setminus\gamma)$,
	   we have: $u\psi\in\dom(\cD_{\bA}^{(\min)})$.
	   
	   Moreover, so we have
	    \[
	     \dom(\cD_{\max}) = \{\psi \in \mathcal{A}_S: 
		\left.\psi\right|_{S_+} = e^{-2\pi i \alpha} \left.\psi\right|_{S_-} \in H_{\loc}^{1/2}( S\setminus \gamma)^2\},
	    \] 
	    and any two elements $\psi,\phi \in \dom(\cD_{\max})$ satisfy
	\begin{align}\label{eq:dmax_stokes}
		\cip{\cD_{\max}\psi}{\phi}_{L^2} - \cip{\psi}{\cD_{\max}\phi}_{L^2}
		= \lim_{\eps \to 0}  \int_{\partial B_{\eps}[\gamma]} \iota_{\d \rho^{\sharp}} \left[\cip{i\bsigma(\d \rho)\psi}{\phi} \vol_{g_3}\right].
	\end{align}
	\end{lemma}
	\begin{proof}
	We use Proposition~\ref{prop:gen_stokes}; testing $\psi \in \dom(\cD_{\max})$ against 
	any $\phi \in \dom(\cD_{\min})$ with $\supp \phi \cap \gamma = \emptyset$. For $\eps>0$ small enough
	in \eqref{eq:general_stokes}, the boundary term over
	$\partial B_{\eps}[\gamma]$ vanishes. To make the remaining boundary terms vanish (for any such $\phi\in \dom(\cD_{\min})$), 
	$\psi$ needs to have the same phase jump across the surface $S$ as the elements in $\dom(\cD_{\min})$.
	Pick now $u\in \sD(\S^3\setminus\gamma)$ localized around the interior of $S$ (that is such that the projection onto $S$
	is defined on $\supp u$ with values in the interior of $S$). The surface $S$ then splits $\supp u$ into two, the part $\omega_+$
	above $S$ and $\omega_-$ below $S$. Define:
	\[
	 \wt{\psi}:= e^{2i\pi\alpha}u\psi\mathds{1}_{\omega_+}+u\psi\mathds{1}_{\omega_-}.
	\]
	By testing against smooth functions in $\sD(\S^3)^2$, and using
	Stokes' formula, we obtain $\bsigma(-i\nabla)\wt{\psi}\in L^2(\S^3)^2$, hence $\wt{\psi}\in H^1(\S^3)^2$.
	For almost all $\bp\in \supp\, u$, we have $|\nabla \wt{\psi}(\bp)|=|\nabla (u\psi)|$, which shows that
	$u\psi\in \dom(\cD_{\min})$ and $\psi|_{S_{\pm}}\in H^{1/2}_{\loc}(S\setminus\gamma)^2$. 
	The extension to all $u \in \sD(\S^3\setminus\gamma)$ is obvious.
	
	Reciprocally let $\psi \in \mathcal{A}_S$ with traces $\psi|_{S_{\pm}}\in H_{\loc}^{1/2}( S\setminus \gamma)^2$
	satisfying the phase jump condition. Let $\phi\in\dom(\cD_{\min})$
	with $\supp\,\phi\cap\gamma=\emptyset$. By Proposition~\ref{prop:gen_stokes}, we have:
	\[
	 \cip{\bsigma(-i\nabla)\psi|_{\Omega_S}}{\phi}_{L^2(\S^3)^2}-\cip{\psi}{\cD_{\min}\phi}_{L^2(\S^3)^2}=0.
	\]
	By density in the graph norm, this holds for all $\phi\in\dom(\cD_{\min})$ and $\psi\in\dom(\cD_{\max})$.
	
	At last, \eqref{eq:dmax_stokes} follows from \eqref{eq:general_stokes}: as $\eps\to 0$, the left-hand side
	of \eqref{eq:general_stokes} converges to that of \eqref{eq:dmax_stokes} by dominated convergence, hence
	the right-hand side also converges giving \eqref{eq:dmax_stokes}.
	\end{proof}

	We conclude the analysis of the maximal operator with a technical lemma that will be useful several times later on.
	\begin{lemma}\label{lem:reg_rho_dmax}
		Given $\psi \in \dom(\cD_{\max})$ with $\supp \psi \subset B_{\delta}[\gamma]$, 
		then we have $\rho\psi\in\dom(\cD_{\min})$ and $\rho \nabla_X\big|_{\Oms}\psi \in L^2(\Omega_S)^2$ for any smooth vectorfield $X$.
	\end{lemma}
	\begin{proof}
		Let $\chi\in\sD(\R,[0,1])$ with $\supp\,\chi\in [-2^{-1},2^{-1}]$ and $\eta=1-\chi$.
		By dominated convergence we can write
		$$
			\rho\psi = \lim_{\eps \to 0}\rho \eta(\rho/\eps)\psi,
		$$
		where the limit is understood in the graph norm, and by the previous lemmma
		$\rho \eta(\rho/\eps)\psi \in \dom(\cD_{\min})$.
		By \eqref{eq:dirac_gradient} we know that $\rho\psi \in H^1(\Omega_S)^2$. Furthermore, for any 
		vector field $X$ we have on $\Omega_S$ that
		$$
			\rho\nabla_X\psi = \nabla_X(\rho\psi) - X(\rho)\psi.
		$$
		Since $X(\rho)$ is bounded, the claim follows.
	\end{proof}
	
	\medskip
	\noindent\textit{Symmetry:}
	Let $\psi,\phi \in \dom\big(\cD_{\bA}^{(\pm)}\big)$. Let us first assume that both
	$\xi_{\mp}^*(\chi_{\eps_0,\gamma}\psi)\xi_{\mp}$ and $\xi_{\mp}^*(\chi_{\eps_0,\gamma}\phi)\xi_{\mp}$ have support away from $\gamma$.
	From \eqref{eq:rel_phase} and \eqref{eq:drho} it follows that
	\begin{equation}\label{eq:sigma_drho}
		\begin{pmatrix}\cip{\xi_+}{\bsigma(\d \rho)\xi_+} & \cip{\xi_+}{\bsigma(\d \rho)\xi_-} \\
		\cip{\xi_-}{\bsigma(\d \rho)\xi_+} & \cip{\xi_-}{\bsigma(\d \rho)\xi_-}\end{pmatrix}
		= \begin{pmatrix} 0 & e^{-i\theta} \\e^{i\theta} & 0\end{pmatrix},
	\end{equation}
	and the boundary terms in \eqref{eq:dmax_stokes} vanish for $\eps$ small enough. Since the left-hand side
	of \eqref{eq:dmax_stokes} is continuous in the graph norm, the general case follows by approximation
	of a sequence of elements who's spin down (resp. spin up) component vanishes on $\gamma$.
	
	\medskip
	\noindent\textit{Self-adjointness:}
	For the proof of the self-adjointess we need an alternative characterization of $\dom(\cD_{\max})$ in terms
	of the coordinates $(s,\rho,\theta)$ and the model case discussed in Section~\ref{ssec:T_straight_line}.
	
	We pick any $\psi \in \dom(\cD_{\max})$. After localizing with $\chi_{\delta,\gamma}$,
	we can assume that $\supp \psi \subset B_{\delta}[\gamma]$.
	We decompose $\psi$ as 
	\begin{align}\label{def:f_pm}
		\psi(\bp) &= f(s(\bp),\rho(\bp),\theta(\bp)) \cdot \xi(\bp)\nn\\
		&= f_{+}(s(\bp),\rho(\bp),\theta(\bp))\xi_{+}(\bp) + f_{-}(s(\bp),\rho(\bp),\theta(\bp))\xi_{-}(\bp)
	\end{align}
	and regard $f$ from now on as a function of $(s,\rho,\theta)$. 
	A computation shows that on $B_{2\delta}[\gamma] \setminus S^c$
	we have
	\begin{align}\label{eq:D_max_close_knot}
		\begin{pmatrix}\cip{\xi_+}{\cD_{\max}\psi} \\ \cip{\xi_-}{\cD_{\max}\psi}\end{pmatrix}(\bp)
		= \big(\wtd + \mathcal{E}_1 + \mathcal{E}_0 \big) f(s(\bp),\rho(\bp),\theta(\bp))
	\end{align}
	with
	\begin{equation}\label{eq:D_gamma_in_T}
		\left\{
		\begin{array}{rcl}
			\wtd&=& -i\begin{pmatrix} \partial_s & e^{-i\theta}(\partial_{\rho} - \frac{i}{\rho}\partial_{\theta})\\
			e^{i\theta}(\partial_{\rho} + \frac{i}{\rho}\partial_{\theta}) & -\partial_s\end{pmatrix},\\
			\cE_1 &=& -i \begin{pmatrix} (h^{-1}-1)\partial_s & \frac{-ie^{-i\theta}(\rho-\sin\rho)}{\rho \sin\rho}\partial_{\theta}\\
			\frac{ie^{i\theta}(\rho-\sin\rho)}{\rho \sin\rho}\partial_{\theta} & -(h^{-1}-1)\partial_s\end{pmatrix}
		 	+ i\sigma_3 \tfrac{\tau_r}{h} \partial_\theta,\\
			\cE_0 &=& -i(\sigma_3 M_{\xi}(\bT) + \sigma_1M_{\xi}(\bS) + \sigma_2M_{\xi}(\bN)).
		\end{array}\right.
	\end{equation}
	Here, all partial derivatives are acting on the set $O_{S}$, which is defined as the image
	of $B_{2\delta}[\gamma] \cap \Omega_S$ under the $(s,\rho,\theta)$-coordinate map.
	
	The general strategy is now as follows. We will show that $\cE_0f$ and $(\cE_1f)\big|_{\theta\neq 0}$ are square integrable,
	hence $\psi$ is in $\dom(\cD_{\max})$ if and only if $(\wtd f)\big|_{\theta\neq 0}$ is square integrable, that is:
	$f \in \dom(\cD_{\T_{\ell},\alpha}^{(\max)})$.
	\begin{rem}
	Note that on $\T \times B_{\delta}[0;\R^2]$,
	being square integrable with respect to the volume form \eqref{eq:pullb_volform} is \textit{equivalent} to being
	square integrable with respect to the standard volume form $\rho\,\d s\wedge \d\rho \wedge \d\theta$.
	\end{rem}

	\begin{lemma}
		Given $\psi \in \dom(\cD_{\max})$ with $\supp \psi \subset B_{\delta}[\gamma]$, 
		$\cE_0f$ and $(\cE_1f)_{|_{\theta\neq 0}}$ are square integrable.
	\end{lemma}
	\begin{proof}
		The statement for $\cE_0$ is evidently true, since the operator does not contain any differential operators
		and the involved matrices are bounded, hence the expression is square integrable.
		
		For $\cE_1$ we proceed as follows.
		\begin{align*}
			\cE_1
			= -i\sigma_3\left((h^{-1}-1)\partial_s - \tfrac{\tau_r}{h}\partial_{\theta}\right)
			-i(-\sin\theta \sigma_1 + \cos\theta \sigma_2)\frac{\rho-\sin\rho}{\rho\sin\rho}\partial_{\theta}.
		\end{align*}
		By introducing an additional order-zero term (which is again square integrable) we obtain the expression
		\begin{align}\label{eq:E1_comp_1}
			&\big(-i\sigma_3\big[(h^{-1}-1)\partial_s - \tfrac{\tau_r}{h}\partial_{\theta}\nn\\
			&\quad +\left((1-h)M_{\xi}(\bT)-\tau_r\sin\rho M_{\xi}(\bG)\right)\big]f_{|_{\theta\neq 0}}\big)(s(\bp),\rho(\bp),\theta(\bp))\nn\\
			&\quad = \begin{pmatrix}\cip{\xi_+}{-i\sigma(\bT^{\flat})[(1-h)\nabla_{\bT} - \tau_r\sin\rho\nabla_{\bG}]\psi_{|_{\Oms}}}\\
			\cip{\xi_-}{-i\sigma(\bT^{\flat})[(1-h)\nabla_{\bT} - \tau_r\sin\rho\nabla_{\bG}]\psi_{\Oms}}
			\end{pmatrix}(\bp),
		\end{align}
		which is square integrable by Lemma~\ref{lem:reg_rho_dmax}, since $1-h(\bp) = \mathcal{O}(\rho)$.
		We recall that $\bG\in\Gamma(\rT\S^3)$ is the vector field defined by the requirement that
		$(\bT,(\d\rho)^{\sharp},\bG)$ is orthonormal \eqref{def:bg}, and that the vectorfield 
		$(\tfrac{1}{h_n}-1)\partial_{s_n}-\tfrac{\tau_n}{h_n}\partial_{\theta_n}$ is the pushforward of 
		$(1-h_n)\bT^{(n)}-\tau_n\sin(\rho_n)\bG^{(n)}$ through the coordinate map $\bp\mapsto (s,\theta,\rho)(\bp)$.
		Similarly $\tfrac{\rho_n-\sin(\rho_n)}{\rho_n\sin(\rho_n)}\partial_{\theta_n}$ 
		corresponds to $\tfrac{\rho_n-\sin(\rho_n)}{\rho_n}\bG^{(n)}$ and we have:
		\begin{multline}\label{eq:E1_comp_2}
			\big(-i(-\sin\theta \sigma_1 + \cos\theta \sigma_2)\Big[\frac{\rho-\sin\rho}{\rho\sin\rho}\partial_{\theta}
			+\tfrac{\rho-\sin\rho}{\rho}M_{\xi}(\bG)\Big]f\big)_{|_{\theta\neq 0}}(s(\bp),\rho(\bp),\theta(\bp))\\
			= \begin{pmatrix}\cip{\xi_+}{-i\sigma(\bG^{\flat})\tfrac{\rho-\sin\rho}{\rho}\nabla_{\bG}\psi_{|_{\Oms}}}\\
			\cip{\xi_-}{-i\sigma(\bG^{\flat})\tfrac{\rho-\sin\rho}{\rho}\nabla_{\bG}\psi_{|_{\Oms}}}
			\end{pmatrix}(\bp),
		\end{multline}
		which is again square integrable by Lemma~\ref{lem:reg_rho_dmax}.
	\end{proof}
	
	Through a local gauge transformation we may assume that
	$$
		S \cap B_{\delta}[\gamma] = \{\bp \in B_{\delta}[\gamma]: \theta(\bp)=0\}.
	$$
	The action of $\wtd$ on $f$ in \eqref{eq:D_gamma_in_T} 
	is then exactly the one of the model operator
	$\cD_{\T_{\ell}, \alpha}^{(\max)}$ in the singular gauge (recall in particular the 
	explanation in Remark~\ref{rem:dmod_singg}), thus
	$f \in \dom(\cD_{\T_{\ell}, \alpha}^{(\max)})$ as a function of $(s,\rho,\theta)$.
	
	Using \eqref{eq:sigma_drho} in \eqref{eq:dmax_stokes} and
	decomposing $\phi$ as $\phi = g_{+}\xi_{+} + g_{-}\xi_{-}$
	one readily obtains
	\begin{align*}
		&\cip{\cD_{\max}\psi}{\phi}_{L^2} - \cip{\psi}{\cD_{\max}\phi}_{L^2}\\
		&\quad = -i\lim_{\delta \to 0} \int_{\partial B_{\delta}[\gamma]} 
		\left[e^{i\theta}\overline{f_{-}} g_{+} + e^{-i\theta}\overline{f_{+}} g_{-}\right] \bG^{\flat} \wedge \bT^{\flat}\\
		&\quad = -i\lim_{\delta \to 0} \int_{s=0}^{\ell} \int_{\theta=0}^{2\pi} h(s,\delta,\theta)\sin(\delta) \d s \d\theta
		\left[e^{i\theta}\overline{f_{-}}g_{+} + e^{-i\theta}\overline{f_{+}}g_{-}\right],
	\end{align*}
	where we used that $\d\rho \wedge \bG^{\flat} \wedge \bT^{\flat} = \vol_{g_3} 
	= h(s,\rho,\theta)\sin(\rho)\,\d s\wedge \d\rho \wedge \d\theta$.
	
	\begin{lemma}\label{lem:sa_knot_tech}
		We have
		\begin{multline*}
			-i\lim_{\delta \to 0} \int_{0}^{\ell} \int_{0}^{2\pi} h(s,\delta,\theta)\sin(\delta)
			\left[e^{i\theta}\overline{f_{-}}g_{+} + e^{-i\theta}\overline{f_{+}}g_{-}\right]\, \d\theta\d s\\
			=-i\lim_{\delta \to 0} \int_{0}^{\ell} \int_{0}^{2\pi}
			\delta \left[e^{i\theta}\overline{f_{-}}g_{+} + e^{-i\theta}\overline{f_{+}}g_{-}\right]\,\d\theta\d s.
		\end{multline*}
	\end{lemma}
	\noindent With the above Lemma (whose proof is relegated to the Appendix) 
	and \eqref{eq:dtmax_ibp} we see that
	\begin{multline*}
		\cip{\cD_{\max}\psi}{\phi}_{L^2(\S^3)^2} - \cip{\psi}{\cD_{\max}\phi}_{L^2(\S^3)^2}\\
		= \cip{\cD_{\T_{\ell},\alpha}^{(\max)}f}{g}_{L^2(\T \times \R^2)^2} 
		- \cip{f}{\cD_{\T_{\ell},\alpha}^{(\max)}g}_{L^2(\T \times \R^2)^2}.
	\end{multline*}
	
	We now have that
	$$
		\psi \in \dom\big(\cD_{\bA}^{(\pm)}\big) \quad \textrm{ if and only if } \quad 
		f \in \dom\big(\cD_{\T_{\ell},\alpha}^{(\pm)}\big).
	$$
	This follows from the fact that on $\Omega_S$, for any vector field $X$ we have
	$$
		\nabla_{X}(f_{\pm}\xi_{\pm}) = X(f_{\pm}) + f_{\pm}\nabla_{X}\xi_{\pm},
	$$
	implying that $f_{\pm}\xi_{\pm} \in H^1(\Omega_S)^2$ if and only if
	$f_{\pm}$ is a $H^1$-function of $(s,\rho,\theta) \in \T_{\ell} \times B_{\delta}[0]$, meaning
	$(f_+,0)$ (resp $(0,f_-)$) is in the domain of $\cD_{\T,\alpha}^{(\min)}$.
	Therefore the self-adjointness of $\dom\big(\cD_{\T,\alpha}^{(\pm)}\big)$ implies the 
	self-adjointness of $\cD_{\bA}^{(\pm)}$.

\subsection{Generalization to a magnetic link}\label{sec:gen_mag_link}
\subsubsection{The singular gauge for a magnetic link}
For $K \in \N$, consider a $K$-link $\gamma = \bigcup_{k=1}^K \gamma_k \subset \S^3$. 
For each knot $\gamma_k$ we pick a Seifert surface $S_k$ with $\partial S_k = \gamma_k$
and set $\bA_k=2\pi\alpha_k\, [S_k]$, $\alpha_k \in \Tf$.
We may then work with the gauge 
potential\footnote{Observe that the formulation in the Seifert gauge with a \textit{common} Seifert surface
is impossible when the fluxes on different boundary curves do not agree.}
\begin{equation}\label{eq:def_sing_gauge_link}
	\bA =\sum_k \bA_k 
	= \sum_k 2\pi\alpha_k\, [S_k].
\end{equation}
The important difference to the case of a single knot is that now two different $S_k$'s may intersect.
We assume that for all $k\neq k'$, the curve $\gamma_k$ is \emph{transverse} to $S_{k'}$,
meaning
$$
	\forall \bp \in \gamma_k \cap S_{k'}, \quad \rT_{\bp}\gamma_k + \rT_{\bp}S_{k'} = \rT_{\bp}\S^3.
$$
Observe that this is always possible; by a careful analysis of the proof of the
Extension Theorem in \cite[Chapter 2]{GuiPoll74} we can always deform the Seifert surfaces infinitesimally
to obtain transversality. The embedding property is preserved since the class of smooth embeddings
is open in the Whitney topology, see \cite[Thm.~1.4]{Hirsch76}.

We denote the space of smooth knots in $\S^3$ by $\sK$ 
and the set of Seifert surfaces with boundary in $\sK$ by $\sS_{\sK}$. 
Furthermore, we define
$$
	\sS_{\sK}^{(K)} := \{(S_1,\dots,S_K) \in \textstyle{\prod_{k=1}^K \sS_{\sK}} :
	\partial S_i\textrm{ and } S_j \textrm{ are transverse in } \S^3, \, i \neq j\}.
$$

We henceforth use the following notation: to a $K$-link $\gamma = \cup_{k=1}^K\gamma_k$ we assign 
a $K$-tuple of Seifert surfaces 
$$
	\uS := (S_1,\ldots,S_K) \in \sS_{\sK}^{(K)},
$$
together with a $K$-tuple of (renormalized) fluxes 
$$
	\ua := (\alpha_1,\ldots,\alpha_K) \in \Tf^K
$$
describing the individual fluxes on each $\gamma_k$. Furthermore, $\bA$ denotes the singular gauge
\eqref{eq:def_sing_gauge_link}.

\subsubsection{The Dirac operator with a magnetic link}
The domain of the minimal Dirac operator is the closure under the graph norm of 
smooth functions on 
$$
	\Omega_{\uS} := \S^3\setminus \big(\cup_k S_k \big)
$$ 
with support in $\S^3 \setminus \gamma$, 
a well-defined limit on each side of $S_k$ and the correct phase-jump condition across $S_k$ for all $k$.
Again, like in Proposition~\ref{prop:egalite_min_dom}, the domain of the minimal operator can be characterized by
\begin{multline*}
	\dom\big(\cD_{\bA}^{(\min)}\big) =H^{1}_{\bA}(\S^3)^2\\			
	:= \left\{\psi \in H^1(\Omega_{\uS})^2: 
	\left.\psi\right|_{(\wt{S}_k)_+}=e^{-2i\pi\alpha_k}\left.\psi\right|_{(\wt{S}_k)_-}\in H^{1/2}(\wt{S}_k)^2, 1 \leq k \leq K\right\},
\end{multline*}
where 
$$
	\wt{S}_k := S_k \cap \big(\cap_{k \neq k'}\complement_{\S^3}S_{k'}\big).
$$
(See Section~\ref{sec:tech_min_dom} below). The operator acts like the free 
Dirac operator on $\Omega_{\uS}$ and the corresponding elements 
$(\sigma(-i\nabla)\psi)\big|_{\Omega_{\uS}}\in L^2(\Omega_{\uS})^2$ are canonically embedded in $L^2(\S^3)^2$. 
Furthermore, for $\psi\in\dom(\cD_{\bA}^{(\min)})$ we have the energy equality
\[
	\int |\cD_{\bA}^{(\min)}\psi|^2=\int |(\nabla \psi)\big|_{\Omega_{\uS}}|^2+\frac{3}{2}\int |\psi|^2.
\]

Let $\eps>0$ be small enough such that the tubular neighborhoods around each
$\gamma_k$ are mutually disjoint,
$$
	B_{\eps}[\gamma] = \bigcup_{k=1}^K B_{\eps}[\gamma_k],
$$
and that on each $\gamma_k$, the map $\exp$ is a diffeomorphism.
We then define $\delta>0$ as
$$
	0 < \delta < \eps \min\left\{1, \left(\sup_{k,n}\norm{\kappa_k^{(n)}}_{L^{\infty}} 
	+ \sqrt{\eps + \sup_{k,n}\norm{\kappa_k^{(n)}}_{L^{\infty}}}\right)^{-1}\right\},
$$
where $\kappa_k^{(n)}$ denotes:
\[
\kappa_k^{(n)}(s):=\sup_{\theta}|\kappa_{g,k}^{(n)}(s)\cos(\theta)+\kappa_{n,k}^{(n)}(s)\sin(\theta)|.
\]
The operator $\cD_{\bA}^{(\underline{e})}$, $\underline{e} \in \{+,-\}^K$, 
is now defined as in Theorem~\ref{thm:dirac_sa_knot} close to each knot $\gamma_k$.

\begin{rem}
	Note that for each Seifert surface $S_k$, we have a Seifert frame $(\bT_k,\bS_k,\bN_k)$
	and sections $\xi_{k,\pm}$ defined in the vicinity of $\gamma_k = \partial S_k$.
	To simplify notation we drop the subscript $k$ for both the Seifert frame and the sections
	$\xi_{\pm}$.
\end{rem}

\begin{theorem}\label{thm:dirac_sa_link}
	Fix $K \in \N$ and let $\gamma = \bigcup_{k=1}^K \gamma_k \subset \S^3$ be a link.
	Pick $(\uS,\ua) \in \sS_{\sK}^{(K)} \times \Tf^K$ with $\partial S_k = \gamma_k$ and set
	$$
		\bA = \sum_{k=1}^K2\pi\alpha_k[S_k].
	$$
	For $\underline{e} \in \{+,-\}^K$, we define $\cD_{\bA}^{(\underline{e})}$ by
	\begin{equation*}
		\left\{
		\begin{array}{lcl}
		\dom\big(\cD_{\bA}^{(\underline{e})}\big) 
		&:=&\big\{\psi \in \dom\big( (\cD_{\bA}^{(\min)})^{*}\big):\\
		&&\quad \langle \xi_{-e_k},\chi_{\delta,\gamma_k}\psi \rangle \xi_{-e_k} 
		\in \dom\big( \cD_{\bA}^{(\min)}\big), 1 \leq k \leq K\big\},\\
		\cD_{\bA}^{(\underline{e})}\psi &:=& (-i\bsigma(\nabla)\psi)\big|_{\Omega_{\uS}} \in L^2(\Omega_{\uS})^2 \hookrightarrow L^2(\S^3)^2.
		\end{array}
		\right.
	\end{equation*}
	The definition of the operators is independent of the choice of $\chi_{\delta,\gamma}$ and the operators
	are self-adjoint.
\end{theorem}

\begin{rem}\label{rem:d_ext_choice}
	As in the case of one knot, the phase jumps for $2\pi\alpha_k[S_k]$ and
	$2\pi(1-\alpha_k)[-S_k]$ agree. Moreover, the complex line bundles $\C\xi_{\pm}$ 
	for $S_k$ coincide with the complex line bundles $\C \xi_{\mp}$ for $-S_k$.
	By adjusting the fluxes $2\pi\alpha_k$ and reorienting the Seifert surfaces $S_k$ accordingly,
	we may without loss of generality assume that $\underline{e} = (-, \ldots, -)$,
	and write
	$$
		\cD_{\bA}^{(-)} = \cD_{\bA}^{(-,\ldots,-)}.
	$$
\end{rem}

\begin{rem}
	Given $\bA = \sum_{k=1}^K 2\pi\alpha_k [S_k]$, the charge-conjugate potential $\bA_{\mathrm{C}}$ is again defined as
	$$
		\bA_{\mathrm{C}} := \sum_{k=1}^K 2\pi (1-\alpha_k)[S_k].
	$$
	The charge conjugation $\mathrm{C}$ maps $\dom(\cD_{\bA}^{(\underline{e})})$ 
	onto $\dom(\cD_{\bA_{\mathrm{C}}}^{(-\underline{e})})$ and we have
	$$
		\mathrm{C}\, \cD_{\bA}^{(\underline{e})} = \cD_{\bA_{\mathrm{C}}}^{(-\underline{e})}\, \mathrm{C}.
	$$
\end{rem}

Again, these self-adjoint extensions are described by the behavior of the spinors close to each knot $\gamma_k$,
yet the boundary distributions on each $\gamma_k$ may have phase-jumps along $\T_{\ell_k}$
due to intersecting Seifert surfaces. Before giving the proof of the above theorem, we discuss this important
aspect.

\subsubsection{Removing the phase jumps}\label{sec:remov_phase}
Assume that the link $\cup_k \gamma_k$ intersects the surfaces $S_{k'}$'s non trivially.

Up to taking $\eps>0$ small enough we can assume that for all $1\le k\neq k'\le K$,
the two subsets $B_{\eps}[\gamma_k]\cap S_{k'}$ and $\gamma_k\cap S_{k'}$ have the same
number of connected components.

Let $1\le k\neq k'\le K$ with $\gamma_k\subset S_{k'}$: we define a  map $E_{k,k'}$ 
describing the phase jumps on $B_{\eps}[\gamma_k]$ due to $S_{k'}$.

\paragraph{Construction of $E_{k,k'}$}

Let $1\le k\neq k'\le K$. If $\gamma_k\cap S_{k'}=\emptyset$, we set $E_{k,k'}\equiv 1$,
else we proceed as follows.

The curve $\gamma_k$ intersects $S_{k'}$ at the points 
$0\le s_1<s_2<\ldots<s_{M}<\ell_{k}$. Call $C_1,\cdots,C_{M}$ the corresponding connected
components of the intersection $B_{\eps}[\gamma_k]\cap S_{k'}$.

For $1\le m\le M$, $S_{k'}$ induces a phase jump $e^{ib_{m}}$ across the cut $C_m$,
where $b_{m}=\pm 2\pi\alpha_{k'}$.

\paragraph{\textit{Case 1: $M=1$}}
		The surface $S_{k'}$ cuts $\gamma_k$ only once and the cut neighborhood
		$B_\eps[\gamma_k]\setminus S_{k'}=:B(\eps,k,k')$ is contractible. 
		We can thus lift the coordinate map $s_k(\cdot)$ on this subset which gives a smooth function
		\[
		    s_{k,k'}:B_\eps[\gamma_k]\setminus S_{k'}\mapsto \mathbb{R},
		\]
		satisfying for all $\bp\in B_\eps[\gamma_k]\setminus S_{k'}$:
		\[
		    \mathrm{exp}\Big(\frac{2i\pi}{\ell_k}s_{k,k'}(\bp)\Big)=\mathrm{exp}\Big(\frac{2i\pi}{\ell_k}s_{k}(\bp)\Big).
		\]
		We then define for all $\bp\in B_\eps[\gamma_k]\setminus S_{k'}\supset B_\eps[\gamma_k]\cap\Omega_{\uS}$
		\begin{equation}\label{Eq:E_k_k'_1_intersection}
			E_{k,k'}(\bp):=\exp\Big(-i b_{1}\frac{s_{k,k'}(\bp)}{\ell_k}\Big).
		\end{equation}
\paragraph{\textit{Case 2: $M\ge 2$}} 
		In this case the cuts $C_m$'s split $B_{\eps}[\gamma_k]$ into $M$ sections:
		\[
			B_{\eps}[\gamma_k] \cap \complement S_{k'}
			=: R_{12} \cup R_{23} \cup \dots \cup R_{M-1)M}\cup R_{M(M+1)},
		\]
		When passing from $R_{(m-1)m}$ to $R_{m(m+1)}$, we are going through $C_{m}$
		which then induces the phase jump $e^{ib_{m}}$ (with $C_{M+1}=C_1$).
		As in the first case, we pick a smooth lift $s_{k,k'}$ of $s_k$
		defined on $B_\eps[\gamma_k]\setminus C_1$.
		
		Writing for $2\le m\le M$
		$$
			R_{m(M+1)}:=R_{m(m+1)} \cup R_{(m+1)(m+2)} \cup \dots \cup R_{M(M+1)},
		$$
		we set for all $\bp\in B_\eps[\gamma_k]\cap\Omega_{\uS}$:
		\begin{equation}\label{Eq:E_k_k'_many_intersections}
			E_{k,k'}(\bp):=\exp\Big(-i\sum_{m=1}^M b_m\frac{s_{k,k'}(\bp)}{\ell_k}
				+i \sum_{m=2}^{M}b_m \mathds{1}_{R_{m(M+1)}}(\bp)\Big).
		\end{equation}
		
	Note that
	    \[
		\sum_{m=1}^{M}b_{m}=-2\pi\alpha_{k'}\link(\gamma_k,\gamma_{k'}),
	     \]
	as the linking number $\link(\gamma_k,\gamma_{k'})$ \cite{Rolfsen}*{Part D, Chapter 5} 
	corresponds to the number of algebraic crossing of $\gamma_{k}$ through the Seifert surface $S_{k'}$ for $\gamma_{k'}$.

Accompanying these decompositions, for $\bp \in B_{\eps}[\gamma_k] \cap \Omega_{\uS}$ we define
the function
\begin{equation}\label{def:curve_phasej_rem}
	E_k(\bp) :=
	\prod_{k'\neq k}E_{k,k'}(\bp),
\end{equation}
and introduce the slope
\begin{equation}\label{eq:pjump_linkn}
	c_k:=\sum_{k'\neq k}2\pi\alpha_{k'}\link(\gamma_{k'},\gamma_{k}).
\end{equation}
Note that the function $E_k$ has a bounded derivative in $B_{\eps}[\gamma_k] \cap \Omega_{\uS}$, and that for any vector field $X$,
$\overline{E_k}X(E_k)=ic_k X(s_k(\cdot))$ can be extended to an element in $C^1(B_{\eps}[\gamma_k])$. Here $s_k(\cdot)$ denotes the coordinate map 
that associates to a point $\bp\in B_\eps[\gamma_k]$ the arclength parameter $s_k(p)\in\R/(\ell_k\Z)$ of its projection onto $\gamma_k$.

The $E_{k,k'}$'s are $\S^1$-valued functions, locally depending only on $s_k$,
with a fixed slope and the correct phase jump across the cuts $C_m$'s: 
they are a unique up to a constant phase.

\begin{proposition}\label{prop:sa_link}
	Let $(\,\cdot\,) = \{(\max),(-),(\min)\}$, then the map $\psi \to \overline{E_k}\psi$ maps
	the set $\{\psi \in \dom(\cD_{\bA}^{(\,\cdot\,)}): \supp\psi \in B_{\eps}[\gamma_k] \cap \Omega_{\uS}\}$
	onto the set $\{\psi \in \dom(\cD_{\bA_k}^{(\,\cdot\,)}): \supp\psi \in B_{\eps}[\gamma_k] \cap \Omega_{\uS}\}$.
	And for $\psi \in \dom(\cD_{\bA}^{(\,\cdot\,)})$ localized around $\gamma$, we have:
	\[
	 \cD_{\bA}^{(\,\cdot\,)}\psi=E_k\cD_{\bA}^{(\,\cdot\,)}(\overline{E}_k\psi)+\frac{c_k}{h_k}\sigma(\bT_k^{\flat})\psi.
	\]
\end{proposition}
\begin{proof}
	The proof is straight forward and left to the reader.
\end{proof}

\subsubsection{Proof of Theorem~\ref{thm:dirac_sa_link}}
As in the proof of the case of the knot we set
$$
	\cA_{\uS} := \big\{\psi \in [\sD'(\Omega_{\uS})]^2 : -i\bsigma(\nabla)\psi \in L^2(\Omega_{\uS})^2\big\} \cap L^2(\Omega_{\uS})^2,
$$
and one can then extend \eqref{eq:general_stokes} to the case of a link.
After testing against elements from the minimal domain, one
arrives at the following characterization of the maximal domain:
\begin{align*}
	\dom\big(\cD_{\bA}^{(\max)}\big) &:= \dom\big( (\cD_{\bA}^{(\min)})^{*}\big)\\
	&= \{\psi \in \cA_{\uS}: -i\bsigma(\nabla)\psi \in L^2(\Omega_{\uS})^2,\\
	&\qquad \left.\psi\right|_{(\wt{S}_k)_+}=e^{-2i\pi\alpha_k}\left.\psi\right|_{(\wt{S}_k)_-}\in 
	H_{\loc}^{1/2}(\wt{S}_k \setminus \gamma_k)^2, 1 \leq k \leq K\}.
\end{align*}
Thus, for $\psi, \phi \in \dom\big(\cD_{\bA}^{(\max)}\big)$ one has
\begin{equation*}
	\cip{\cD_{\bA}^{(\max)}\psi}{\phi}_{L^2} - \cip{\psi}{\cD_{\bA}^{(\max)}\phi}_{L^2}
	= \lim_{\delta \to 0^{+}}\sum_{k=1}^K\int_{\partial B_{\delta}[\gamma_k]} \iota_{(\d \rho_k)^{\sharp}}
	\left[\cip{i\bsigma(\d \rho_k)\psi}{\phi} \vol_{g_3}\right].
\end{equation*}
From this it immediately follows that the operators $\cD_{\bA}^{(\underline{e})}$ are symmetric.
Furthermore, as
$$
	\cip{i\bsigma(\d \rho_k)\psi}{\phi} = \cip{i\bsigma(\d \rho_k)\overline{E}_k\psi}{\overline{E}_k\phi},
$$
the self-adjointness follows from Proposition~\ref{prop:sa_link} and Theorem~\ref{thm:dirac_sa_knot}.

\subsection{Remarks on the Pauli operator}
From the previous study it is then natural to define the Pauli operator associated 
to the gauge potential $\bA:=\sum_{k=1}^K2\pi\alpha_k[S_k]$ as the square of 
the Dirac operator $\cD_{\bA}^{(-)}$. Equivalently, it is the unbounded operator 
associated to the quadratic form
\[
	q_{\bA}[\psi]:= \int|\cD_{\bA}^{(-)}\psi|^2
\]
with form domain $\dom\big(\cD_{\bA}^{(-)}\big)$.
In a forthcoming paper \cite{dirac_s3_paper3}, we will prove that $\cD_{\bA}^{(-)}$ is the limit 
(in particular in the strong resolvent sense) of Dirac operators with smooth magnetic fields. 
In that sense, the quadratic form $q_{\bA}$ gives rise to the correct Pauli operator.

\section{Convergence \& compactness properties of Dirac operators}
In this section we are interested in describing the change in the spectrum of the
operators $\cD_{\bA}^{(\underline{e})}$ when
varying the fluxes and deforming the link $\gamma$.
We will restrict our analysis to the extension $\cD_{\bA}^{(-)}$, results for
$\cD_{\bA}^{(\underline{e})}$ follow via Remark~\ref{rem:d_ext_choice}.
Furthermore, we denote the graph norm of $\cD_{\bA}^{(-)}$ by
$$
	\norm{\psi}_{\bA}^2 := \cip{\psi}{\psi}_{L^2} + \cip{\cD_{\bA}^{(-)}\psi}{\cD_{\bA}^{(-)}\psi}_{L^2},
	\quad \psi \in \dom\big(\cD_{\bA}^{(-)}\big).
$$

In order to specify what we mean by deforming a link, it is necessary to introduce
a notion to compare sub-manifolds of $\S^3$.

\subsection{Metrics for knots and Seifert surfaces in $\C^2$}
We begin by defining a metric to compare smooth compact oriented sub-manifolds of the same 
dimension in $\C^2$. This metric should in some sense be able to detect convergence of the differential structures of
the sub-manifolds and can therefore be seen as a smooth version of the Hausdorff distance 
for compact sets.

Let $\cM_1$ and $\cM_2$ be two smooth oriented compact sub-manifolds of $\C^2$ of the same real dimension $d$, then
denote for $i=1,2$
\begin{align*}
	N_i: 
	\begin{array}{ccl}
		\cM_i &\longrightarrow &\widetilde{\mathrm{Gr}}_{4-d,4}\\
		\bp &\mapsto &(\rT_p\cM_i)^{\bot},
	\end{array}
\end{align*}
their respective Gauss maps. Observe that the oriented Grassmanian
$\widetilde{\mathrm{Gr}}_{4-d,4}$ can be canonically embedded into 
the Grassmann algebra $\bigwedge^{4-d}\C^2$ (over $\R$).
We extend the differential $\d N_i(\bp)$, initially only defined on $\rT_p\cM_i$,
to all of $\C^2$ by
\begin{align*}
	\d N_i(\bp): 
	\begin{array}{ccl}
		\rT_p\cM_i \oplus (\rT_p\cM_i)^{\bot} &\longrightarrow& \textstyle{\bigwedge^{4-d}\C^2}\\
		\bv + \bv_{\bot} &\mapsto& \d N_i(\bp)\bv.
	\end{array}
\end{align*}
Similarly, we can extend $\d^kN_i(\bp)$ to a map
\begin{align*}
	\d^kN_i: 
	\begin{array}{ccl}
		(\C^2)^k &\longrightarrow& \textstyle{\bigwedge^{4-d}\C^2}\\
		(\bv_1,\dots,\bv_k) &\mapsto& \d^kN_i(\bp)[\bv_1,\dots,\bv_k].
	\end{array}
\end{align*}
Furthermore, we set
$$
	\norm{\d^kN_i(\bp)}^2 := \sup_{(\bv_1,\dots,\bv_k), |\bv_j|^2=1}\left|\d^kN_i(\bp)[\bv_1,\dots,\bv_k]\right|^2,
$$
and define the distance between two points $\bp_1\in \cM_1, \bp_2 \in \cM_2$ as
$$
	\delta_{\cM_1\times\cM_2}(\bp_1,\bp_2)
	:= |\bp_1-\bp_2| + \sum_{k=0}^\infty 2^{-k}\min\left\{\norm{\d^kN_1(\bp_1) - \d^kN_2(\bp_2)},1\right\}.
$$
We then introduce the following distance on oriented $d$-dimensional sub-manifolds:
\begin{align}\label{def:dist_subm}
	&\dist_d(\cM_1,\cM_2) := |\cH_d(\cM_1)-\cH_d(\cM_2)|\\
	&\quad +\max\left(\sup_{\bp_1 \in \cM_1}\inf_{\bp_2 \in \cM_2}\delta_{\cM_1\times\cM_2}(\bp_1,\bp_2),
	\sup_{\bp_2 \in \cM_2}\inf_{\bp_1 \in \cM_1}\delta_{\cM_1\times\cM_2}(\bp_1,\bp_2)\right).\nn
\end{align}
Here, $\cH_d$ is the $d$-dimensional Hausdorff measure.


In our situation we need to compare Seifert surfaces; for two such surfaces $S_1,S_2$, we will use the metric
$$
	\dist_{\sS}(S_1,S_2) := \dist_2(S_1,S_2) + \dist_1(\partial S_1, \partial S_2).
$$
A natural metric on $\sS_{\sK}^{(K)}\times \Tf^K$ is given by
\begin{equation}\label{def:dist_metric}
	\dist((\uS,\ua),(\uS',\ua'))
	:= \max_{1\leq k\leq K}\left[\dist_{\sS}(S_k,S_k') + \dist_{\Tf}(\alpha_k,\alpha_k')\right].
\end{equation}
We recall that $\dist_{\Tf}$ denotes the distance \eqref{def:dist_torus}.

Due to the geometrical nature of the coordinates, the following holds.
\begin{proposition}[Convergence of coordinates]\label{prop:conv_coord}
	Let $S^{(n)}$ be a sequence of Seifert surfaces converging to $S$. Then
	the geodesic coordinates $(s_n,\rho_n,\theta_n)$ for $S^{(n)}$ converge to the geodesic
	coordinates on $S$ in the $C^\infty$-norm, with $\tfrac{\ell}{\ell_n}s_n \to s$  on $\overline{B}_{\eps}[\partial S]$
	and $\rho_n \to \rho$,$\theta_n \to \theta$ on 
	$\overline{B}_{\eps}[\partial S] \setminus B_{\eps'}[\partial S]$ for $0<\eps'<\eps$.
	Furthermore $\rho_n \to \rho$ converges in the $C^0$-norm on $\overline{B}_{\eps}[\partial S]$.
\end{proposition}
\noindent For the above proposition we chose a base point $\bp_0\in \partial S$ 
and picked $\bp_0^{(n)}\in\partial S^{(n)}$ such that
$$
	\delta_{S\times S^{(n)}}(\bp_0,\bp_0^{(n)})=\inf_{\bp\in\partial S^{(n)}}\,\delta_{S\times S^{(n)}}(\bp_0,\bp).
$$
The convergence of $s_n$ holds as it corresponds to the projection onto the curves $\gamma_n$. 
This convergence is strong, as the metric on Seifert surface measures the variations of the differential 
structure of the surfaces.
Similarly, the convergence of the function $\rho_n$ follows from the fact that it corresponds to 
the (geodesic) distance from the curves.
The convergence of the angle $\theta_n$ is a consequence of the convergence of $s_n,\rho_n$
and the convergence of the Seifert frame 
$$
	(\bT(\gamma^{(n)}(s_n)),\bS(\gamma^{(n)}(s_n)),\bN(\gamma^{(n)}(s_n)))
$$
(viewed as an orthonormal family in $\C^2$), where we have identified 
$\gamma^{(n)}:=\partial S^{(n)}$ with its arc length parametrization with base point
$\bp_0^{(n)}$.

\subsection{Convergence of Dirac operators}
We now discuss several results concerning the convergence of Dirac operators.
\begin{theorem}[Compactness]\label{thm:compactness}
	Fix $K \in \N$. Let
	$$
		(\uS^{(n)},\ua^{(n)}),(\uS,\ua) \in \sS_{\sK}^{(K)} \times \Tf^K
	$$
	such that $(\uS^{(n)},\ua^{(n)}) \to (\uS,\ua)$ in the $\dist$-metric \eqref{def:dist_metric}, 
	and set
	$$
		\bA^{(n)} = \sum_{k=1}^K 2\pi \alpha_k^{(n)} [S_k^{(n)}],
		\quad \bA = \sum_{k=1}^K 2\pi \alpha_k [S_k].
	$$
	Furthermore, let $(\psi^{(n)})_{n \in \N}$ be a sequence in $L^2(\S^3)^2$ such that:
	\begin{enumerate}
		\item $\psi^{(n)} \in \dom\big(\cD_{\bA^{(n)}}^{(-)}\big)$, $n \in \N$.
		
		\item $\big(\norm{\psi^{(n)}}_{\bA^{(n)}}\big)_{n \in \N}$ is uniformly bounded.
	\end{enumerate}
	Then, up to extraction of a subsequence, $\psi^{(n)} \rightharpoonup \psi \in \dom\big(\cD_{\bA}^{(-)}\big)$
	and
	$$
		\big(\psi^{(n)}, \cD_{\bA^{(n)}}^{(-)}\psi^{(n)}\big)
		\rightharpoonup \big(\psi,\cD_{\bA}^{(-)} \psi\big)
		\textrm{ in } L^2(\S^3)^2 \times L^2(\S^3)^2
	$$
	with
	$$
		\int \big|\cD_{\bA}^{(-)} \psi \big|^2
		\leq \liminf_{n\to+\infty} \int \big|\cD_{\bA^{(n)}}^{(-)} \psi^{(n)} \big|^2.
	$$
	In fact, one has strong convergence except in the cases when $\alpha_k^{(n)} \to 1^{-}$,
	where the loss of $L^2$-mass of $(\psi^{(n)})_n$ can only occur 
	through concentration onto such a knot
	$\gamma_k = \partial S_k$.
\end{theorem}

\begin{theorem}[Strong resolvent convergence]\label{thm:str_res_cont}
	Fix $K \in \N$. Let
	$$
		(\uS^{(n)},\ua^{(n)}),(\uS,\ua) \in \sS_{\sK}^{(K)} \times \Tf^K
	$$
	such that $(\uS^{(n)},\ua^{(n)}) \to (\uS,\ua)$ in the $\dist$-metric \eqref{def:dist_metric}, 
	and set
	$$
		\bA^{(n)} = \sum_{k=1}^K 2\pi \alpha_k^{(n)} [S_k^{(n)}],
		\quad \bA = \sum_{k=1}^K 2\pi \alpha_k [S_k].
	$$
	Then $\cD_{\bA^{(n)}}^{(-)}$ converges to $\cD_{\bA}^{(-)}$ in the strong resolvent sense.
\end{theorem}

An important consequence of Theorem~\ref{thm:compactness} is the following result.

\begin{theorem}\label{thm:discrete_spectrum}
	For $(\uS,\ua) \in \sS_{\sK}^{(K)} \times \Tf^K$, the spectrum of $\cD_{\bA}^{(-)}$ is discrete and 
	the operator is Fredholm.
\end{theorem}
\begin{proof}
	To prove the first statement, we use Weyl's criterion \cite[Theorem~VII.12]{ReedSimon1}
	for the characterization of the essential spectrum:
	$\lambda \in \spec_{\ess}(\cD_{\bA}^{(-)})$ if and only if there exists
	a $L^2$-normalized sequence $(\psi^{(n)})_n$ such that
	$$
		\cip{\psi^{(n)}}{\psi^{(m)}}_{L^2} = \delta_{n,m},
	$$
	where $\delta_{n,m}$ is the Kronecker symbol, with
	$$
		\lim_{n\to\infty}\norm{(\cD_{\bA}^{(n)}-\lambda)\psi^{(n)}}_{L^2} = 0.
	$$
	For the operator $\cD_{\bA}^{(-)}$, this is impossible by Theorem~\ref{thm:compactness}, 
	since any sequence $(\psi^{(n)})_{n}$ in $\dom\big(\cD_{\bA}^{(-)}\big)$ 
	which is $\norm{\,\cdot\,}_{\bA}$-bounded is $L^2(\S^3)^2$-convergent up to extraction.
	
	By the same argument the kernel of $\cD_{\bA}^{(-)}$ is finite dimensional. Furthermore,
	from the Banach closed range theorem and the self-adjointness it follows that
	$$
		\coker\big(\cD_{\bA}^{(-)}\big) \cong \ker\big(D_{\bA}^{(-)}\big),
	$$
	and the operator is thus Fredholm.
\end{proof}


\begin{corollary}[Convergence of eigenfunctions]\label{cor:corres_spectrum}
	Fix $K \in \N$, and let $\uS^{(n)} \to \uS$ in $\sS_{\sK}^{(K)}$,
	$\ua^{(n)} \to \ua$ in $[0,1)^K$. 
	If $(\psi^{(n)})_{n \in \N}$ is a sequence of $L^2$-normalized functions 
	with 
	$$
		\psi^{(n)} \in \ker\big(\cD_{\bA^{(n)}}^{(-)} -\lambda^{(n)}\big)
	$$ 
	and $\lim_{n\to+\infty}\lambda^{(n)}=\lambda$,
	then up to extraction of a subsequence, $(\psi^{(n)})_{n \in \N}$ converges to a function
	$$
		\psi \in \ker\big(\cD_{\bA}^{(-)} - \lambda\big)
	$$ 
	in $L^2(\S^3)^2$.
\end{corollary}

Armed with the above, it is now simple to prove a non-existence for zero modes for small values of $\ua$.
When $\ua = 0$, we have $\cD_{\bA}^{(-)} = -i\bsigma(\nabla)$ by Proposition~\ref{prop:dmin_alpha_0},
and from \eqref{eq:dirac_gradient} we know that the kernel of $-i\bsigma(\nabla)$ is trivial.
Combining this with a contradiction argument involving Corollary~\ref{cor:corres_spectrum} gives the following result.
\begin{corollary}[Non-existence of zero modes for small fluxes]\label{coro:no_zero_modes_for_small_fluxes}
	Fix $K\in\N$, let $\gamma \subset \S^3$ be any $K$-link with 
	Seifert surface $\uS\in\sS_{\sK}^{(K)}$, and set 
	$\bA=\sum_{k=1}^{K}2\pi\alpha_k[S_k]$. 
	Then there exists $0<\alpha_0(\gamma)<1$ such that for any $\ua\in[0, \alpha_0(\gamma)]^K$, 
	the kernel $\ker\cD_{\bA}^{(-)}$ is trivial.	
\end{corollary}

\subsection{Proofs of Theorems~\ref{thm:compactness}~\&~\ref{thm:str_res_cont}}
We now make some preparations before giving the proofs of the main theorems of this section.
\subsubsection{Localization}\label{sec:loc}
Let $\eps>0$ such that $B_{\eps}[\gamma_k] \cap B_{\eps}[\gamma_{k'}] = \emptyset$ for $k \neq k'$
and the coordinates $(s_k,\rho_k,\theta_k)$ are well-defined on $B_{\eps}[\gamma_k]$.
Note that by the continuity of the coordinates (see Proposition~\ref{prop:conv_coord}), 
the same will be true for for the links $\gamma^{(n)}$ for $n$ large enough.

For $1 \leq k \leq K$, recall the function $\chi_{\delta,\gamma_k^{(n)}}$ 
as in \eqref{eq:chi_loc_curve}, with 
$$
	0<\delta< \eps \min\left\{1, \left(\sup_{k,n}\norm{\kappa_k^{(n)}}_{L^{\infty}} 
	+ \sqrt{\eps + \sup_{k,n}\norm{\kappa_k^{(n)}}_{L^{\infty}}}\right)^{-1}\right\}.
$$
Furthermore, set
$$
	\chi_{\delta,S_k^{(n)}}(\bp) 
	:= \chi\left(4\tfrac{\dist_{g_3}(\bp,S_k^{(n)})}{\delta}\right)\left(1-\chi_{\delta,\gamma_{k}^{(n)}}(\bp)\right),
$$
which has support close to the Seifert surface $S_k^{(n)}$, but away from the knot $\gamma_k^{(n)}$.
The remainder is then defined as
$$
	\chi_{\delta,R_k^{(n)}}(\bp) := 1-\chi_{\delta,S_k^{(n)}}(\bp) - \chi_{\delta,\gamma_k^{(n)}}(\bp),
$$
which constitutes a partition of unity subordinate to $S_k^{(n)}$.
The partition of unity for the entire link is then given by
\begin{align}\label{def:part_unity}
	1 &= \prod_{k=1}^K\left(\chi_{\delta,\gamma_k^{(n)}}(\bp) + \chi_{\delta,S_k^{(n)}}(\bp)
	+\chi_{\delta,R_k^{(n)}}(\bp) \right)\nn\\
	&= \sum_{\underline{a} \in \{1,2,3\}^K}\chi_{\delta,\underline{a}}^{(n)}(\bp)
	=\sum_{k=1}^K\chi_{\delta,\gamma_k^{(n)}}(\bp)+\sum_{\underline{a} \in \{2,3\}^K}\chi_{\delta,\underline{a}}^{(n)}(\bp),
\end{align}
where $\chi_{\delta,\underline{a}}^{(n)}$ is the product $\prod_{k=1}^K\chi_{\delta,X_k^{(n)}}$, with 
$$
	X_k^{(n)} = \left\{\begin{array}{lr}
	\gamma_k^{(n)}, &a_k = 1,\\
	S_{k}^{(n)}, &a_k = 2,\\
	R_k^{(n)}, &a_k = 3.
	\end{array}\right.
$$

\subsubsection{Proof of Theorem~\ref{thm:compactness}}\label{sec:proof_compact}
	To keep notation simple, we will drop the $(-)$ superscript for $\cD_{\bA}^{(-)}$.

	By the Banach-Alaoglu theorem the sequences converge weakly,
	$$
		\psi^{(n)} \rightharpoonup \psi, \quad \cD_{\bA^{(n)}}\psi^{(n)} \rightharpoonup w \textrm{ in } L^2(\S^3)^2.
	$$
	We will first show that $(\psi^{(n)})_n$ converges strongly to a $\psi$ in $L^2(\S^3)^2$ and 
	$\psi$ is non-zero (except when $\alpha_k^{(n)} \to 1^{-}$); the closedness is shown
	at the very end.
	To do so, we decompose $\psi^{(n)}$ with respect to the partition of unity \eqref{def:part_unity} and
	analyze each group separately.
	
	\medskip
	\noindent\textit{Away from the knots:}
	This corresponds to the functions coming from the last term in \eqref{def:part_unity}.
	Fix $\underline{a} \in \{2,3\}^K$, and write $\{1,\dots,K\} = J_2 \cup J_3$ where $a_k=2, k \in J_2$ and 
	$a_k=3, k \in J_3$.
	If $J_2$ is empty, then the sequence $(\chi_{\delta,\underline{a}}^{(n)}\psi^{(n)})_n$ has no phase jump 
	and is $H^1(\S^3)^2$-bounded. It thus converges strongly in $L^2(\S^3)^2$ up to extraction.
	
	Assume now that $J_2$ contains at least one element. 
	For each $1 \leq k \leq K$, the Seifert surface $S_k^{(n)}$ splits $\supp \chi_{\delta,S_{k}^{(n)}}$ into two 
	disconnected regions $O_{k,\pm}^{(n)}$, which correspond to the regions above and below the Seifert surface.
	We remove the phase-jump with the help of the function
	\begin{equation}\label{def:n_phase_remove}
		E_{\delta,\underline{a}}^{(n)}:=\prod_{k \in J_2}\exp\left[2\pi i \alpha_k^{(n)}\mathds{1}_{O_{k,-}^{(n)}}\right].
	\end{equation}
	The function $\overline{E}_{\delta,\underline{a}}^{(n)}\chi_{\delta,\underline{a}}^{(n)}\psi_n$ then contains no phase jump and
	is again $H^1(\S^3)^2$-bounded, hence converges up to extraction of a subsequence.

	\subsubsection*{Close to the knots:}
	We now turn to the sequence $(\chi_{\delta,\gamma_k^{(n)}}\psi^{(n)})_n$.
	We need the following lemma, whose proof is relegated to the Appendix.
	\begin{lemma}\label{lem:control_Ngr}
		Let $\gamma$ be a knot of length $\ell$ with Seifert surface $S$,
		$0<\alpha<1$ and $\delta>0$ small enough. Given any 
		$\psi\in \dom\big(\mathcal{D}_{\bA}^{(-)}\big)$,
		we write
		$$
			\chi_{\delta,\gamma}\psi=f \cdot \xi = f_+\xi_++ f_{-}\xi_{-}.
		$$
		Then there exists a constant $C=C(\delta,\alpha,S)$ such that
		$$
			\norm{f}_{\T_{\ell}} 
			\leq C\norm{\chi_{\delta,\gamma}\psi}_{\bA}
			\leq C(1+\delta^{-1})\norm{\psi}_{\bA},
		$$
		where $\norm{\,\cdot\,}_{\T_{\ell}}, \norm{\,\cdot\,}_{\bA}$ denote the graph 
		norms of the operators $\cD_{\T_{\ell},\alpha}^{(-)}, \cD_{\bA}^{(-)}$, respectively.
		Moreover, the bound $C(\delta,\alpha,S)$ depends continuously on $S\in\sS_{\sK}$.
	\end{lemma}

	From Proposition~\ref{prop:sa_link} we know that 
	$$
		\phi^{(n)}:=\overline{E}_{k}^{(n)}\chi_{\delta,\gamma_k^{(n)}}\psi^{(n)}
		\in \dom(\cD_{\bA_k^{(n)}}).
	$$
	Furthermore, notice that the derivative of $E_{k}^{(n)}$ is uniformly bounded on its domain, which implies that
	the sequence $(\norm{\phi^{(n)}}_{\bA_k^{(n)}})_{n\in\N}$ is also bounded.
	It therefore suffices to prove the convergence of the sequence $(\phi^{(n)})_n$,
	which will imply the convergence of $\chi_{\delta,\gamma_k^{(n)}}\psi^{(n)}$.
		
	Recall that up to gauge transformation we can assume that the jump phase 
	occurs across $B_\delta[\gamma_k^{(n)}]\cap\{ \theta_k^{(n)}=0\}$.
	For simplicity we will drop the dependence in $k$, and will denote
	$$
		\ell_n := \ell_{k}^{(n)}, \quad \T_n := \R/(\ell_{k}^{(n)}\Z).
	$$
	We write 
	\begin{align}\label{eq:ccurve_decomp_1}
		\phi^{(n)}=f^{(n)} \cdot \xi^{(n)} = f_{+}^{(n)}\xi_+^{(n)} + f_{-}^{(n)}\xi_-^{(n)}
	\end{align}
	and now use Lemma~\ref{lem:control_Ngr} to see that
	the sequence $(\norm{f^{(n)}}_{\T_n})_{n\in\N}$ 
	is also bounded. 
	In $L^2(\T_n\times\R^2;\rho \d s\d \rho\d \theta)^2$, we further decompose
	$f^{(n)}$ into its regular and singular part according to \eqref{eq:model_decomp}, 
	\begin{align}\label{eq:ccurve_decomp_2}
		f^{(n)}= f_0^{(n)} + f_{\sing}^{(n)}	.
	\end{align}
	Observe that if $\alpha_k^{(n)} = 0 \mod 1$, then $f_{\sing}^{(n)} \equiv 0$.
	
	To show that the regular part $f_0^{(n)}$ converges, we use the map
	\begin{equation*}
		\sim\,: 
		\begin{array}{ccl}
			L^2(\T_n \times \R^2)^2 &\longrightarrow& L^2(\T \times \R^2)^2\\
			f(s_n,\rho,\theta) &\mapsto& \wt{f}(s,\rho,\theta) := f(s \ell_n/\ell,\rho,\theta).
		\end{array}
	\end{equation*}
	Recall also that
	$$
		\norm{f_0^{(n)}}_{\T_n}^2 = \int |f_0^{(n)}|^2 + \int_{\T_n \times (\R^2 \setminus \{\theta = 0\})}|\nabla f_0^{(n)}|^2,
	$$
	so the sequence $(\wt{f}_0^{(n)})_n$ is $H^1(\T_n \times (\R^2 \setminus \overline{\{\theta = 0\}}))^2$-bounded,
	and thus converges up to extraction to $\wt{f}_0$.
	
	Recall that $f^{(n)}_{\sing} = f^{(n)}_{\sing}(\ul^{(n)})$, so that with
	\eqref{eq:sing_contr}--\eqref{eq:sing_psi_la} we have
	\begin{align*}
		\norm{f_{\sing}^{(n)}}_{L^2}^2 =C_{\alpha_k^{(n)}}\sum_{j \in \T_{n}^*}\frac{|\lambda_j^{(n)}|^2}{\langle j\rangle^2},
		\quad \norm{f_{\sing}^{(n)}}_{\T_n}^2 =(C_{\alpha_k^{(n)}}+C_{1-\alpha_k^{(n)}})\sum_{j \in \T_{n}^*}|\lambda_j^{(n)}|^2,
	\end{align*}
	with 
	$$
		C_{\alpha} := \int_0^{\infty}|K_{\alpha}(r)|^2\,rdr.
	$$ 
	We now distinguish between two different cases.
	
	\smallskip
	\noindent\textit{Case $\lim_{n \to \infty} \alpha_k^{(n)} <1$:}
	We will show that $(\wt{f}_{\sing}^{(n)})_n$ converges in $L^2(\T\times\R^2)^2$. 
	Observe that $\wt{f}_{\sing}^{(n)}$ is again in the singular sector of
	$\dom\big(\cD_{\T,\alpha_k^{(n)}}^{(-)}\big)$ by the very definition
	of the singular subspace.
	Furthermore, the sequence $(\ul^{(n)})_n$ is uniformly bounded in $\ell_2(\T^*)$ by \eqref{eq:sing_contr},
	hence by the Banach-Alaoglu Theorem we can extract a subsequence (also denoted by $(\ul^{(n)})_n$), converging
	in $\hat{H}^{-1}(\T_{\ell})$, whose limit is $0$ when $\alpha_k^{(n)} \to 0$, 
	because $\lim_{x \to 1^{-}}(1-x)C_x > 0$. This implies that $(\wt{f}_{\sing}^{(n)})_n$ converges
	to 
	$$
		\wt{f}_{\sing}(s,\rho,\theta) = \frac{1}{\sqrt{2\pi \ell}}\sum_{j \in \T^*}\lambda_j e^{ijs}\begin{pmatrix} 0 \\ 
		K_{\alpha}(\rho \langle j\rangle)e^{i\alpha \theta}\end{pmatrix}.
	$$
	
	Now set
	\begin{equation*}
		\begin{array}{ll}
			\phi_{0}^{(n)} = \chi_{2\delta,\gamma_k^{(n)}}(f_{0}^{(n)}\cdot \xi^{(n)}),
			&\phi_{\sing}^{(n)} = \chi_{2\delta,\gamma_k^{(n)}}(f_{\sing}^{(n)}\cdot \xi^{(n)}),\\
			\phi_{0} =\chi_{2\delta,\gamma_k}(\wt{f}_{0}\cdot \xi),
			&\phi_{\sing} = \chi_{2\delta,\gamma_k}(\wt{f}_{\sing}\cdot \xi).
		\end{array}
	\end{equation*}
	It remains to show that $\phi_n = \phi_0^{(n)} + \phi_{\sing}^{(n)}$ converges to
	$\phi = \phi_0 + \phi_{\sing}$. 
	We only show how to do it for the singular part, the calculation for the regular part is the same.
	We begin with
	$$
		\chi_{2\delta,\gamma_k^{(n)}}(f_{\sing}^{(n)} \cdot \xi^{(n)})
		= \chi_{2\delta,\gamma_k}(f_{\sing}^{(n)} \cdot \xi^{(n)}) 
		+ (\chi_{2\delta,\gamma_k^{(n)}} -  \chi_{2\delta,\gamma_k})(f_{\sing}^{(n)} \cdot \xi^{(n)}).
	$$
	The last term converges to $0$ in $L^2(\S^3)^2$ since 
	$$
		\limsup_{n \to \infty}\norm{\chi_{2\delta,\gamma_k^{(n)}} -  \chi_{2\delta,\gamma_k}}_{L^{\infty}} = 0.
	$$
	For the first term we now write
	\begin{align*}
		\chi_{2\delta,\gamma_k}(f_{\sing}^{(n)} \cdot \xi^{(n)}) 
		&= \chi_{2\delta,\gamma_k}(\wt{f}_{\sing} \cdot \xi)
		+ \chi_{2\delta,\gamma_k}(f_{\sing}^{(n)} - \wt{f}_{\sing})\cdot \xi^{(n)}\\
		&\quad +\chi_{2\delta,\gamma_k}f_{\sing}\cdot (\xi^{(n)} - \xi).
	\end{align*}
	Again, the last term converges to $0$ since
	$$
		\limsup_{n \to \infty} \norm{\chi_{2\delta,\gamma_k}(\xi^{(n)} - \xi)}_{L^{\infty}} = 0.
	$$
	
	It remains to show that $\chi_{2\delta,\gamma_k}(f_{\sing}^{(n)} - \wt{f}_{\sing})\cdot \xi^{(n)} \to 0$ 
	in $L^2(\S^3)^2$, and this is done in coordinates.
	For $\bp \in \S^3$,
	\begin{align*}
		(f_{\sing}^{(n)} - \wt{f}_{\sing})(\bp) 
		&= f_{\sing}^{(n)}\big((s_n,\rho_n,\theta_n)(\bp)\big) - \wt{f}_{\sing}\big((s,\rho,\theta)(\bp)\big)\\
		&= \left[f_{\sing}^{(n)}\big((s_n,\rho_n,\theta_n)(\bp)\big) - \wt{f}_{\sing}\big((\tfrac{\ell}{\ell_n}s_n,\rho_n,\theta_n)(\bp)\big)\right]\\
		&\quad+ \left[\wt{f}_{\sing}\big((\tfrac{\ell}{\ell_n}s_n,\rho_n,\theta_n)(\bp)\big) - \wt{f}_{\sing}\big((s,\rho,\theta)(\bp)\big)\right]\\
		&=: \delta_1^{(n)}f_{\sing}(\bp) + \delta_2^{(n)}f_{\sing}(\bp).
	\end{align*}
	For $n$ large enough,
	$$
		\norm{\chi_{\delta,\gamma_k}\delta_1^{(n)}f_{\sing}}_{L^2(\S^3)^2}^2
		\leq C(\delta) \tfrac{\ell_n}{\ell}\norm{\wt{f}_{\sing}^{(n)}-\wt{f}_{\sing}}_{L^2(\T \times \R^2)^2}^2,
	$$
	which tends to zero when $n \to \infty$.
	It remains to treat $\delta_2^{(n)}f_{\sing}(\bp)$. If $\wt{f}_{\sing}$ were smooth, the convergence would follow
	immediately from the convergence of the coordinates. In the general case we can however use a standard mollifying 
	argument (with mollifiers supported in $[-\ell/4,\ell/4]\times B_1(0) \subset \T \times \R^2$)
	to return to the smooth case.
	
	\smallskip
	\noindent\textit{Case $\lim_{n \to \infty} \alpha_k^{(n)} = 1^{-}$:}
	We once again decompose 
	$$
		f^{(n)}= f_0^{(n)} + f_{\sing}^{(n)}(\ul^{(n)}).
	$$
	Just as in the previous case, one can show that the regular part $f_0^{(n)}$ converges (up to extraction),
	whereas for the singular part we only know that
	$$
		\limsup_{n \to \infty} (1-\alpha_k^{(n)})\norm{\ul^{(n)}}_{\ell_2(\T^*)} < \infty.
	$$
	Furthermore, by testing $\wt{f}_{\sing}^{(n)}(\ul^{(n)})$ against a Hilbert basis of
	$$
		L^2(\T;ds)\otimes L^2(\R^2;\rho\d\rho\d\theta)^2=L^2(\T\times\R^2)^2,
	$$
	it is clear that this function tends weakly to $0$ in $L^2(\T\times\R^2)^2$ and concentrates around $\T\times\{0\}$.
	On $L^2(\S^3)^2$, this implies that $\phi_{\sing}^{(n)} \rightharpoonup 0$ 
	and concentrates around $\gamma_k^{(n)}$, which could lead to a loss of $L^2$-mass.
	
	\medskip
	\noindent\textit{Closedness:}
	It remains to show that $w = \cD_{\bA}\psi$.
	Firstly, by testing against any $\phi \in C_0^{\infty}(\Omega_{\uS})^2$ we get that
	$$
		\cip{\cD_{\bA^{(n)}}\psi^{(n)}}{\phi}_{L^2}
		= \cip{\psi^{(n)}}{-i\bsigma(\nabla)\phi}_{L^2},
	$$
	therefore $\psi \in \mathcal{A}_{\uS}$ with $w = -i\bsigma(\nabla)\psi$. It remains to show that $\psi$ has the 
	correct phase-jump in order for it to be in $\dom(\cD_{\bA})$.
	Away from the link this follows from the a.e.-convergence of $E_{\delta,\underline{a}}^{(n)}\chi_{\delta,\underline{a}}^{(n)}$
	to $E_{\delta,\underline{a}}\chi_{\delta,\underline{a}}$. By a similar argument the same holds for the regular part close
	to the link; it converges to an element in the minimal domain of $\cD_{\bA}$.
	As for the singular part, it is obvious since we know the explicit form of the limiting function.
	These two last facts imply that $\psi\in \dom(\cD_{\bA})$.
	
\subsubsection{Proof of Theorem~\ref{thm:str_res_cont}}
	For this proof, the second characterization of strong resolvent continuity in the following Lemma
	(whose proof can be found in the Appendix) will be the most convenient.
	\begin{lemma}\label{lem:char_sres_conv}
		Let $(\cD_n)_n$ be a sequence of (unbounded) self-adjoint operators on a separable 
		Hilbert space $\cH$. 
		Then the following statements are equivalent.
		\begin{enumerate}
    			\item $\cD_n$ converges to $\cD$ in the strong resolvent sense.
    	
			\item For any $(f,\cD f)\in \cG_{\cD}$,
  			there exists a sequence $(f_n,\cD_n f_n)\in \cG_{\cD_n}$ converging to 
			$(f,\cD f)$ in $\cH \times \cH$.
    	
			\item The orthogonal projection $P_n$ onto $\cG_{\cD_n}$ 
			converges in the strong operator topology to $P$, the orthogonal projector onto $\cG_{\cD}$.
		\end{enumerate}
	\end{lemma}
	It thus suffices to prove the following.
	For any $\psi \in \dom(\cD_{\bA})$ and any sequence $(\uS^{(n)},\ua^{(n)})_n$ converging 
	to $(\uS,\ua)$, there exists a sequence $(\psi^{(n)})_n$ such that
	\[
		\big(\psi^{(n)},\cD_{\bA^{(n)}}\psi^{(n)}\big)
		\to \big(\psi, \cD_{\bA} \psi\big) \textrm{ in } L^2(\S^3)^2 \times L^2(\S^3)^2.
	\]
	We reuse the partition of unity from \eqref{def:part_unity} and localize $\psi$ accordingly.

	\medskip
	\noindent\textit{Away from the knots:}
	This corresponds to the functions coming from the last term in \eqref{def:part_unity}.
	Recall the function $E_{\delta, \underline{a}}^{(n)}$ from \eqref{def:n_phase_remove}; we then
	define the function
	$$
		\psi_{\delta,\underline{a}}^{(n)} 
		:= E_{\delta,\underline{a}}^{(n)}\overline{E}_{\delta,\underline{a}} \psi \in \dom(\cD_{\bA^{(n)}}).
	$$
	Now observe that for a.e. $\bp \in \S^3$ we have
	$$
		\cD_{\bA^{(n)}} \psi_{\delta,\underline{a}}^{(n)}(\bp)
		= E_{\delta,\underline{a}}^{(n)}(\bp)\overline{E}_{\delta,\underline{a}}(\bp)
		\cD_{\bA}(\chi_{\delta,\underline{a}}\psi)(\bp)
	$$
	hence
	$$
		|\psi_{\delta,\underline{a}}^{(n)}(\bp)| = |\chi_{\delta,\underline{a}}\psi(\bp)|,
		\quad |\cD_{\bA^{(n)}}\psi_{\delta,\underline{a}}^{(n)}(\bp)|
		= |\cD_{\bA}(\chi_{\delta,\underline{a}}\psi)(\bp)|,
	$$
	and dominated convergence yields
	$$
		\big(\psi_{\delta,\underline{a}}^{(n)},\cD_{\bA^{(n)}}\psi_{\delta,\underline{a}}^{(n)}\big)
		\to \big(\chi_{\delta,\underline{a}}\psi,\cD_{\bA}(\chi_{\delta,\underline{a}}\psi)\big)
		\textrm{ in } L^2(\S^3)^2 \times L^2(\S^3)^2.
	$$
	
	\medskip
	\noindent\textit{Close to the knots:}  
	We recall $E_k$ for $\gamma_k$ from \eqref{def:curve_phasej_rem};
	then, 
	$$
		\phi_{\delta,k} := \overline{E}_k\chi_{\delta,\gamma_k}\psi \in \dom(\cD_{\bA_k}).
	$$
	Again, for a.e. $\bp \in \S^3$ we have
	\begin{align*}
		\cD_{\bA_k}\phi_{\delta,k}(\bp)
		= \overline{E}_k(\bp)\big(\cD_{\bA}+i\bsigma(\d E_k(\bp))\big)(\chi_{\delta,\gamma_k}\psi)(\bp).
	\end{align*}
	Using \eqref{eq:partiald_coord} and \eqref{eq:pjump_linkn}, for $\bp\in\supp(\chi_{\delta,\gamma_k})\cap\Omega_{\uS}$ we have:
	\begin{multline}\label{eq:dirac_on_E_k}
		i\overline{E}_k\bsigma(\d E_k)(\bp)\\
		=\frac{2\pi}{\ell_k}\sum_{k'\neq k}\alpha_{k'}\link(\gamma_{k'},\gamma_k)
		\bsigma(h\bT^{\flat}+\sin(\rho)\tau_r(s)\bG^{\flat})(\bp),
	\end{multline}
	and a similar formula holds for $E_k^{(n)}$ ($\bG$ is defined in \eqref{def:bg}). By geometric convergence of $\uS^{(n)}$, 
	the sequence $(i\chi_{\delta,\gamma_k^{(n)}}\overline{E}_k^{(n)}\bsigma(\d E_k^{(n)})(\bp))_n$ 
	converges to $i\chi_{\delta,\gamma_k}\overline{E}_k\bsigma(\d E_k)(\bp)$ for a.e. $\bp\in\S^3$.
	It then suffices to approximate $\phi_{\delta,k}$ by elements $\phi_{\delta,k}^{(n)} \in \dom(\cD_{\bA_k^{(n)}})$;
	we then obtain
	$$
		\big(E_k^{(n)}\phi_{\delta,k}^{(n)},\,\cD_{\bA^{(n)}}(E_k^{(n)}\phi_{\delta,k}^{(n)}))
		\to \big(\chi_{\delta,\underline{a}}\psi,\,\cD_{\bA}(\chi_{\delta,\underline{a}}\psi))
		\textrm{ in } L^2(\S^3)^2 \times L^2(\S^3)^2
	$$
	by dominated convergence.
	
	Up to a gauge transformation $e^{iF_k(\bp)}$ we may assume that the phase jump of 
	$\chi_{\delta,\gamma_k}\psi$
	occurs across $\{\theta(\bp)=0\}$. We then decompose $\chi_{\delta,\gamma_k}\psi=f_{\delta,k}\cdot\xi$ as in 
	\eqref{eq:ccurve_decomp_1}. The function $f_{\delta,k}$ is in $\dom(\cD_{\T_{\ell_k},\alpha_k}^{(-)})$ 
	as a function of $(s,\rho,\theta)$ and we split it into its regular and singular part 
	$$
		f_{\delta,k} = f_{\delta,k,0} + f_{\delta,k,\sing}(\ul)
	$$
	as in \eqref{eq:ccurve_decomp_1}.
	For the sake of simplicity we drop the dependence in $\delta$ and $k$ in the functions $f_{\delta,k,(\cdot)}$
	and write $\T_n:=\T_{\ell_k^{(n)}}$.
	
	We form the functions $\wt{f}_{0}^{(n)}$ 
	defined on $L^2(\T_n \times\R^2)^2$ by the formula:
	$$
		\wt{f}_{0}^{(n)}(s,\rho,\theta):=e^{i(\alpha_{k}^{(n)}-\alpha_k)\theta}
			f_{0}\big(\tfrac{\ell_k}{\ell_k^{(n)}}s,\rho,\theta \big)\in\dom(\cD_{\T_n,\alpha_k^{(n)}}^{(\min)})
	$$
	Similarly, using the explicit expression of the singular part, we define the function
	$\wt{f}_{\sing}^{(n)}\in\dom(\cD_{\T_n,\alpha_k^{(n)}}^{(-)})$ by the formula:
	$$
		\wt{f}_{\sing}^{(n)}(s,\rho,\theta):=\frac{1}{\sqrt{2\pi\ell_k^{(n)}}}
			\sum_{j\in\T_n^*}e^{ijs}\lambda_{j(\ell_k^{(n)}/\ell_k)}
			\begin{pmatrix}0 \\ K_{\alpha_k^{(n)}}(\rho\langle j\rangle)e^{i\alpha_k^{(n)}\theta}\end{pmatrix}.
	$$
	We finally set 
	$$
		\wt{f}^{(n)}_{\delta,k}(s,\rho,\theta):=\chi((2\delta)^{-1}\rho)(\wt{f}_{0}^{(n)}+\wt{f}_{\sing}^{(n)})(s,\rho,\theta),
	$$
	where $\chi$ is the function in the definition of $\chi_{\delta,\gamma_k}$.
	
	Coming back to $L^2(\S^3)^2$, the function 
	$$
		\wt{\phi}_{\delta,k}^{(n)}(\bp):=\wt{f}^{(n)}_{\delta,k}(s_n,\rho_n,\theta_n)\cdot\xi^{(n)}(\bp)
	$$ 
	is almost the candidate we are looking for; 
	through a gauge transformation $e^{-i\zeta_k^{(n)}(\bp)}$ we shift the phase jump of $\wt{\phi}_{\delta,k}^{(n)}$
	from $\{ \theta_n(\bp)=0\}$ to $S_k^{(n)}$ and we obtain the function 
	$$
		\phi_{\delta,k}^{(n)}:=e^{-i\zeta_k^{(n)}}\wt{\phi}_{\delta,k}^{(n)} \in \dom\big(\cD_{\bA_k^{(n)}}\big).
	$$
	Following the proof of Theorem~\ref{thm:compactness}, we obtain that
	$\phi_{\delta,k}^{(n)}$ converges to $\chi_{\delta,\gamma_k}\psi$ in $L^2(\S^3)^2$ and that 
	$$
		\cD_{\bA_k^{(n)}}\phi_{\delta,k}^{(n)} 
		\rightharpoonup \cD_{\bA_k}(\chi_{\delta,\gamma_k}\psi) \textrm{ in } L^2(\S^3)^2.
	$$
	With the convergence of the coordinates
	$(s_n,\rho_n,\theta_n)$ and the expression of $-i\bsigma(\nabla)$ in terms of these 
	coordinates \eqref{eq:D_gamma_in_T},
	the same mollifying argument used in the proof of Theorem~\ref{thm:compactness} gives that
	$$
		\cD_{\bA_k^{(n)}}\phi_{\delta,k}^{(n)} \to \cD_{\bA_k}(\chi_{\delta,\gamma_k}\psi)
		\textrm{ in } L^2(\S^3)^2.
	$$

\section{Hopf links}\label{sec:Hopf_links}
\subsection{Non-existence of zero modes for circles}\label{subs:circle}
Here, we will deal with the magnetic field supported on any circle $\cC$ in $\S^3$,
meaning $\cC$ is the intersection of $\S^3$ with any complex line, not including the case
when this intersection is just a point. We then orient this circle and study it with the help of a gauge potential
$\bA = 2\pi \alpha [S]$, $\partial S = \cC$, $\alpha \in \Tf$.
\begin{theorem}\label{thm:no_zero_modes_for_circles}
	Let $\cC$ be an oriented circle in $\S^3$ with Seifert surface $S$, and set $\bA = 2\pi\alpha[S]$.
	For any flux $\alpha \in \Tf$ we have
	$$
		\ker\cD_{\bA}^{(-)}=\{0\}.
	$$
\end{theorem}
\noindent The proof of this theorem is given at the end of this section.

As for Corollary~\ref{coro:no_zero_modes_for_small_fluxes}, a simple contradiction argument 
involving Theorem~\ref{thm:no_zero_modes_for_circles} and Corollary~\ref{cor:corres_spectrum} gives
\begin{corollary}
	Let $\cC \subset \S^3$ be a circle with Seifert surface $S$
	and $\alpha \in \Tf$. Then there exists $\eps=\eps(\cC,\alpha)>0$, such that for all
	$(S',\alpha')$ in the ball $B_{\eps}[(S,\alpha)] \subset \sS_{\sK} \times [0,1)$
	with $\bA' = 2\pi \alpha'[S']$ we have
	$$
		\ker\cD_{\bA'}^{(-)}=\{0\}.
	$$
\end{corollary}

\subsection{The Erd\H{o}s-Solovej construction}
In this section we will derive certain results on zero modes with the help of the construction in \cite{MR1860416}, 
where the authors used the conformal invariance of the kernel of the Dirac operator combined with the
Hopf map to construct a large class of zero modes in $\S^3$.

We will see $\S^2 \subset \R^3$ with the metric $\tfrac{1}{4}g_2$, 
where $g_{2}$ is the induced metric from its ambient space.
We define the Hopf map as
\begin{equation*}
	\Hopf:
	\begin{array}{ccl}
		(\S^3,g_3) &\longrightarrow& (\S^2,\tfrac{1}{4}g_2)\\
		(z_0,z_1) &\mapsto& (|z_0|^2-|z_1|^2,\Re(2z_0\overline{z}_1),\Im(2z_0\overline{z}_1)).
	\end{array}
\end{equation*}
From the geometry of the Hopf map it is clear that the preimage
of $K$ points $v_1, \dots, v_K \in \mathbb{S}^2$ is given by $K$ interlinking 
circles on $\mathbb{S}^3$, $\gamma^K = \Hopf^{-1}(\{v_i\}_{1\le i\le k})$,
with $\link(\gamma_i,\gamma_j) = 1$ for any $i \neq j$.
We then orient $\gamma^K$ along the vector field $u_3 = (iz_0,iz_1)$
and define the \textit{magnetic Hopf link} $\bB^K$ by
\begin{equation}\label{def:magn_hopf_link}
	\bB^K := \sum_{k=1}^K 2\pi \alpha_k \big[\Hopf^{-1}(\{v_k\})\big],
\end{equation}
with (renormalized) flux $\alpha_k \in \Tf$ on each component $\gamma_k$.

The following theorem is essentially \cite[Theorem~8.1]{MR1860416} applied to the case when 
the magnetic field on $\S^2$ is a collection of Dirac points. It is a priori not obvious that one can apply the 
theorem in this more ``singular" context, 
yet a step by step verification of the original proof shows that this is possible.
The intuition behind this fact is that for our self-adjoint extensions, the singular part of the spinor is always pointing in the 
direction of the magnetic field, and is in some sense respects the geometry of the construction, in particular the Hopf
map.

\begin{theorem}\label{thm:erdos-solovej}
	Let $\bB^K$ be a magnetic Hopf link as in \eqref{def:magn_hopf_link} with corresponding singular gauge
	$\bA$ and assume $\alpha_k \in (0,1)$ for all $k$. 
	Let $\cD_{\bA}^{(-)}$ be the Dirac operator on $(\S^3,g_3)$ 
	and define $c \in (-1/2,1/2]$, $m \in \Z$ such that
	$$
		\sum_{k=1}^K \alpha_k =: c + m.
	$$
	Furthermore, set 
	$$
		\beta_{\S^2,k} := \left(\,\sum_{k=1}^K 8\pi\alpha_k \delta_{v_k} - 2(c+k)\right) \frac{\vol_{g_2}}{4}.
	$$
	As $(2\pi)^{-1}\int_{\S^2} \beta_{\S^2,k} = m-k$, on the spinor bundle $\Psi_{m-k}$ (with Chern number $m-k$) there exists a
	two-dimensional Dirac operator $\cD_{\S^2,k}$ with magnetic two-form $\beta_{\S^2,k}$. 
	Then:
	\begin{enumerate}
		\item The spectrum of $\cD_{\bA}^{(-)}$ is given by
		$$
			\spec \cD_{\bA}^{(-)} = \bigcup_{k \in \Z} \left( \mathcal{Z}_k  
			\cup \left\{\pm \sqrt{\lambda^2 + (k+c)^2} - \frac{1}{2} : \lambda \in \spec_+\cD_{\S^2,k}\right\}\right),
		$$
		where
		\begin{align*}
			\mathcal{Z}_k = \left\{\begin{array}{ll}
 				\{k+c-1/2\}, 
				&\quad m>k,\\
				\emptyset,
				&\quad m=k,\\
				\{-k-c-1/2\},
				&\quad m<k.
				\end{array} \right.
		\end{align*}
		
		\item The multiplicity of an eigenvalue equals the number of ways it can be written as 
		$\sqrt{\lambda^2 + (k+c)^2} - \frac{1}{2}$, $k \in \Z$ and $\lambda \in \spec_+\cD_{\S^2,k}$ counted with
		multiplicity, or as an element in $\mathcal{Z}_k$ counted with multiplicity $|m-k|$.
		
		\item The eigenspace of $\cD_{\bA}^{(-)}$ with eigenvalue in $\mathcal{Z}_k$ contains spinors with definite
		spin value $\mathrm{sgn}(m-k)$.

	\end{enumerate}
\end{theorem}

In the special case when $c=1/2$, we have precise information about $\ker \cD_{\bA}^{(-)}$.

\begin{corollary}\label{cor:dim_ker_hopf_link}
	If the magnetic Hopf link $\bB^K$ has fluxes $\ua \in (0,1)^K$ with $\sum_{k=1}^K\alpha_k = m+1/2$,
	then
	$$
		\dim \ker \cD_{\bA}^{(-)} = m.
	$$ 
\end{corollary}

\subsection{Proof of Theorem~\ref{thm:no_zero_modes_for_circles}}
	We will use the invariance of the kernel of the Dirac operator under conformal maps to study the
	problem in $\R^3$ with the help of the stereographic projection.
	Denote by $\cD_{\cC}:= \cD_{\bA}^{(-)}$.
	
	Here, $g_{\R^3}$ denotes the canonical metric of $\R^3$. Furthermore, we
	set
	$$
		\Omega_3(\bx) := \frac{2}{1+|\bx|^2},
	$$
	the conformal factor for the stereographic projection
	\begin{equation*}
		\st:	
		\begin{array}{ccl}
			\mathbb{S}^3 &\longrightarrow& \R^3 \cup \{\infty\}\\
			(z_0,z_1) &\mapsto& \bx = \left(\frac{\Re(z_0)}{1-\Im(z_1)},\frac{\Im(z_0)}{1-\Im(z_1)},\frac{\Re(z_1)}{1-\Im(z_1)} \right)
		\end{array}
	\end{equation*}
	Here we assumed that the point $(0,i) \in \S^3$ is equidistant to any point on $\cC$, and the stereographic
	projection thus maps $\cC$ to a circle $\wt{\cC}$ (with center $0$) in $\R^3$.
	We may furthermore assume that we are working in the gauge $\bA = 2\pi \alpha [S]$, such 
	that $S$ is mapped to the interior $\D_{\wt{\cC}}$ of the circe $\wt{\cC}$ in its ambient plane.

	Note that on the manifold $(\R^3,\Omega_3^2 g_{\R^3})$, the Dirac operator $\cD_{\wt{\cC}}^{\Omega_3}$ 
	is defined in the same way as in Section~\ref{sec:sa_magn_knot};
	away from the surface $S$, we know that the Dirac operator acts as the free Dirac,
	and by \cite[Theorem~4.3]{MR1860416} we know it has the form
	\begin{equation}\label{eq:free_dir_st}
		\cD^{\Omega_3} = -i\Omega_3^{-2}\sigma\cdot \nabla^{\R^3} \Omega_3.
	\end{equation}
	The minimal domain contains smooth functions away from the surface
	with the correct phase jump, and for the extension the singular part of the spinor is 
	(close to the curve) carried by the component $\wt{\xi}_{-}$ aligned against the magnetic field.
	In particular,
	$$
		\dom\big(\cD_{\wt{\cC}}^{\Omega_3}\big) \subset L^2(\R^3; \Omega_3^3 \d \bx)^2.
	$$
	
	By a conformal transformation we can now map an arbitrary circle $\wt{\cC}$ of radius
	$r$ in $\R^3$ onto the unit circle $\S^1 = \S^1 \oplus \{0\} \subset \R^3$, which gives rise to an isomorphism between
	the two kernels:
	\begin{equation*}
		\begin{array}{ccl}
			\ker \cD_{\wt{\cC}}^{\Omega_3} &\overset{\simeq}{\longrightarrow}& \ker \cD_{\S^1}^{\Omega_3}\\
			\psi &\mapsto& \Omega_3^{-1}(\bx)\Omega_3(\tfrac{R^{-1}\bx}{r})U_R \psi(\tfrac{R^{-1}\bx}{r}),
		\end{array}
	\end{equation*}
	where $R \in \SO(3)$ rotates the (oriented) ambient plane of $\wt{\cC}$ onto $\R^2 \oplus \{0\}$, and $U_R \in \SU(2)$
	is given by 
	$$
		U_R \sigma\, \cdot \nu\, U_R^* = \sigma \cdot (R\nu), \quad \nu \in \R^3.
	$$
	We then conformally map the unit circle $\S^1$ onto the 
	straight line $\mathbb{I}_3 = 0 \oplus \R \subset \R^3$
	through
	\begin{align*}
		\iota:(x_1,x_2,x_3)
		&\overset{\st^{-1}}{\mapsto}
		\left(z_0=\frac{2(x_1+ix_2)}{1+|\bx|^2},z_1=\frac{2x_3+i(|\bx|^2-1)}{1+|\bx|^2} \right)
		\overset{\mathrm{sw}}{\mapsto} (z_1,z_0)\nonumber\\
		&\overset{\st}{\mapsto}
		\left(\frac{2x_3}{|\be_2-\bx|^2},\frac{|\bx|^2-1}{|\be_2-\bx|^2},\frac{2x_1}{|\be_2-\bx|^2} \right),
	\end{align*}
	with $\be_2 := (0,1,0)$, which corresponds to a rotation in $\S^3$.
	Observe that the map $\iota$ is a conformal involution, mapping the Riemannian manifold 
	$\R^3 \setminus \S^1$ onto $\R^3 \setminus (\mathbb{I}_3\cup\{\be_2\})$.
	We therefore obtain an operator $\cD_{\mathbb{I}_3}^{\Omega_3}$, whose gauge
	is given by the surface $\Sigma_{-\pi/2} := \{\bx \in \R^3: x_1 = 0, x_2 < 0\}$, and away from
	$\Sigma_{-\pi/2}$ it acts as \eqref{eq:free_dir_st}.
	
	From this operator we can then obtain $\cD_{\mathbb{I}_3}$, the corresponding operator
	on the manifold $(\R^3,g_{\R^3})$.
	The operator $\cD_{\mathbb{I}_3}$ acts like the free Dirac operator $-i\sigma \cdot \nabla^{\R^3}$ away
	from the surface $\Sigma_{-\pi/2}$, with $\dom\big(\cD_{\mathbb{I}_3}\big) \subset L^2(\R^3)^2$.
	To do the spectral analysis of the operator $\cD_{\mathbb{I}_3}$ one can proceed exactly as in the model case
	in Section~\ref{ssec:T_straight_line}, and many of the results remain true 
	\footnote{Note however that the Pontryagin duality is then $\R^* = \R$ instead of $\T_{\ell}^* = \tfrac{2\pi}{\ell}\Z$.}.
	In particular, for any $\psi \in \dom\big(\cD_{\mathbb{I}_3}\big)$,
	\begin{align}\label{eq:dirac_split_line}
		\int |\cD_{\mathbb{I}_3}\psi|^2
		= \int |(\partial_{3}\psi)\big|_{\R^3 \setminus \Sigma_{-\pi/2}}|^2 
		+ \int |(\sigma_{\bot}\cdot \nabla_{\bot}^{\R^3}\psi)\big|_{\R^3 \setminus \Sigma_{-\pi/2}}|^2,
	\end{align}
	with $x_{\bot} = (x_1,x_2)$.

	As $2\Omega_3^{-1}(\bx)=1+|\bx|^2$, any $\psi \in \dom\big(\cD_{\mathbb{I}_3}^{\Omega_3}\big)$ satisfies
	\begin{equation}\label{eq:estim_Dpsi_0}
		\int_{\R^3}\Omega_3^{-1}(\bx)\big|\mathcal{D}_{\mathbb{I}_3}\big(\Omega_3(\bx)\psi(\bx) \big)\big|^2\d\bx
		\ge \frac{1}{2}\int_{\R^3}|\mathcal{D}_{\mathbb{I}_3}(\Omega_3 \psi)|^2\d\bx.
	\end{equation}
	For $\psi \in \ker\big(\cD_{\mathbb{I}_3}^{\Omega_3}\big)$ we define $\psi_0:=\Omega_3 \psi$, so that $\psi_0$ 
	satisfies $\textstyle{\int_{\R^3}\Omega_3|\psi_0|^2<+\infty}$ and
	$$
		-i\sigma \cdot \nabla^{\R^3}\psi_0=0 \quad \textrm{ in } \mathscr{D}'(\R^3 \setminus \Sigma_{-\pi/2})^2.
	$$
	Next, we consider $(\chi_R)_{R>0}$ where $\chi_R(\bx):=\chi_1(\tfrac{\bx}{R})$ 
	for some $ \chi_1\in C_0^{\infty}(\R^3)$. We can assume for instance that $\chi_1$ is a function of 
	$x_3$ and $|x_{\perp}|$ only and that
	$$
			\mathds{1}\big(x_3^2+|x_{\perp}|^2\le 4^{-1}\big)\le \chi_1(\bx)\le \mathds{1}\big(x_3^2+|x_{\perp}|^2\le 1\big).
	$$
	The functions $\chi_R\psi_0$ are clearly in $\dom(\mathcal{D}_{\mathbb{I}_3})$, 
	and by \eqref{eq:dirac_split_line},
	\begin{multline}\label{eq:dirac_split_line_2}
		\int|\cD_{\mathbb{I}_3}(\chi_R\psi_0)|^2
		=\int|(\partial_{x_3}(\chi_R\psi_0))\big|_{\R^3 \setminus \Sigma_{-\pi/2}}|^2\\
		+\int |(\sigma_{\perp}\cdot \nabla_{\perp}^{\R^3}(\chi_R\psi_0))\big|_{\R^3 \setminus \Sigma_{-\pi/2}}|^2.
	\end{multline}
	
	We now need a lemma, whose proof is given in the Appendix.
	\begin{lemma}\label{lem:decomp_lim}
		Taking the limit $R\to+\infty$ in \eqref{eq:dirac_split_line_2} gives
		\begin{multline}\label{eq:decomp_lim}
			\int|(\sigma\cdot\nabla^{\R^3}\psi_0)\big|_{\R^3 \setminus \Sigma_{-\pi/2}}|^2
			=\int|(\partial_{x_3}\psi_0)\big|_{\R^3 \setminus \Sigma_{-\pi/2}}|^2\\
			+\int|(\sigma_{\perp}\cdot \nabla_{\perp}^{\R^3}\psi_0)\big|_{\R^3 \setminus \Sigma_{-\pi/2}}|^2.
		\end{multline}
	\end{lemma}
	
	If $\psi$ is a zero mode for $\cD_{\mathbb{I}_3}^{\Omega_3}$, then the above quantity is zero. 
	This however implies that $\psi_0$ is function of $(x_1,x_2)$ only, and $\psi_0(x_1,x_2)$ is
	a formal zero mode for $-i\sigma_\perp \cdot\nabla_\perp^{\R^3}$, so we have
	$$
		\psi_0(x_1,x_2)=C\begin{pmatrix}0\\ p(x_1-ix_2)(x_1-ix_2)^{-\alpha}\end{pmatrix},
	$$
	where $p$ is a polynomial. Here, the branch cut of the complex logarithm is along $\Sigma_{-\pi/2}$.
	
	The integrability condition $\int |\psi_0(\bx)|^2\frac{\d\bx}{1+|\bx|^2}<+\infty$ is equivalent to
	$$
		\int_{\R^2}|\psi_0(x_1,x_2)|^2\frac{\d x_1\d x_2}{\sqrt{1+|\bx_{\perp}|^2}} <+\infty,
	$$
	and so we can only have $p(x_1-ix_2) = \textrm{const}$ and $2^{-1} < \alpha < 1$.
	This proves that for $\alpha\le 2^{-1}$ there is no zero mode for $\cD_{\mathbb{I}_3}^{\Omega_3}$ 
	and thus $\cD_{\bA}$ also has no zero modes.
	
	There remains to check whether this candidate is a real zero mode for $2^{-1}<\alpha<1$. 
	We will now prove that it is not; although the spinor
	$$
		\psi(\bx) = \Omega_3(\bx)^{-1}\begin{pmatrix}0\\(x_1-ix_2)^{-\alpha}\end{pmatrix}
	$$
	only has a spin down component along the straight line
	it carries in some sense a Dirac point in its spin up component at infinity.
	
	Both $\mathcal{M}_1:=(\R^3\setminus (\mathbb{I}_3 \cup \{\be_2\}),\Omega_3^2g_{\R^3})$ and 
	$\mathcal{M}_2:=(\R^3\setminus\S^1,\Omega_3^2g_{\R^3})$ 
	are two (isometric) charts of $\S^3 \setminus (\st^{-1}(\S^1 \cup \{\infty\})$,
	and we will continue working in the latter.
	The following lemma is proved in the Appendix.
	\begin{lemma}\label{lem:zero_m_s1}
		The corresponding candidate $\wt{\psi}$ in the chart $\R^3 \setminus \S^1$ is given by
		$$
			\wt{\psi}(\bx) = \frac{\Omega_3(\iota(\bx))^{-1} |\bx-\be_2|^{2\alpha}}{(2x_3-i(|\bx|^2-1))^{\alpha}}
			U_{\bx}i\sigma_3 U_{-\iota(\bx)}\begin{pmatrix}0\\ 1 \end{pmatrix}
			\in L^2(\R^3)^2,
		$$
		where
		$$
			U_{\bx}:=\frac{1+i\bsigma\cdot \bx}{\sqrt{1+|\bx|^2}}.
		$$
		The wave function $\wt{\psi}$ satisfies
		$$
			\cD^{\Omega_3}\wt{\psi}=0 \quad \textrm{ in } \mathscr{D}'(\R^3 \setminus \D)^2.
		$$
	\end{lemma}
	
	Let us make the expression of $\wt{\psi}$ slightly more precise:
	for simplicity we write $2x_3+i(|\bx|^2-1)$ in polar coordinates:
	$$
		\wt{\rho}(\bx)e^{i\wt{\theta}(\bx)}:=2x_3+i(|\bx|^2-1).
	$$
	Moreover, a computation shows that $\tfrac{1+|\iota(\bx)|^2}{1+|\bx|^2}=2|\bx-\mathbf{e}_2|^{-2}$, 
	so we obtain the more explicit expression
	\begin{equation}
		\wt{\psi}(x)=\frac{(1+i\bx\cdot\sigma)e^{i\alpha\wt{\theta}(\bx)}}{\sqrt{2}\wt{\rho}(\bx)^\alpha|\bx-\be_2|^{2(1-\alpha)}}
		\begin{pmatrix}\frac{\wt{\rho}(\bx)}{|\bx-\be_2|}e^{-i\wt{\theta}(\bx)}\\ \frac{2x_1}{|\bx-\be_2|}-i|\bx-\be_2| \end{pmatrix}.
	\end{equation}
	To see that $\wt{\psi}$ is not in $\dom(\mathcal{D}_{\S^1}^{\Omega_3})$, 
	it suffices to prove that there exists a $\phi_0\in \dom(\mathcal{D}_{\S^1}^{\Omega_3})$ such that:
	\begin{equation*}
		\langle \cD^{\Omega_3}\phi_0,\,\wt{\psi} \rangle_{L^2}
		-\langle \phi_0,\, \cD^{\Omega_3}\wt{\psi} \rangle_{L^2}
		=\langle \cD^{\Omega_3}\phi_0;\,\wt{\psi} \rangle_{L^2}\neq 0,
	\end{equation*}
	where the $L^2$ in the above scalar product stands for $L^2(\R^3\setminus\D,\Omega_3^3\d\bx)^2$.
	We choose a spinor $\phi_0$ which close to the curve behaves like 
	$$
		u(\bx)U_{\bx}\begin{pmatrix}0 \\ (2x_3-i(|\bx|^2-1))^{-\alpha} \end{pmatrix}, \quad u\in C^\infty_0(\R^3).
	$$
	We now apply Stokes' formula on the complement $\Omega_\eps$
	of the $g_{\R^3}$-tubular neighborhood $B_\eps[\S^1]$;
	taking the limit $\eps\to 0$ shows that
	\[
		i\int_{\partial\Omega_\eps}\langle \sigma\cdot\bn_{\partial\Omega_\eps}\Omega_3\phi_0,\,\Omega_3\wt{\psi} \rangle
		\underset{\eps\to 0}{\longrightarrow}
		-\pi u(\be_2)\int_{-\infty}^{+\infty}\frac{2^{2(1-\alpha)} \d s}{(1+s^2)^{1+2(1-\alpha)}}\not\equiv 0.
	\]
	The proof of the theorem is now complete.

\appendix
\section{}

\subsection{Proof of Proposition~\ref{prop:egalite_min_dom}}\label{sec:tech_min_dom}
Consider the set
\begin{multline*}
	\mathrm{Ans}_{\bA}^{(\min)}:=\big\{\psi\in C^{\infty}(\Omega_{\uS})^2: 
	\supp\psi\subset\S^3\setminus\gamma,\\
	\left.\psi\right|_{(S_k)_\pm} \textrm{ exist \& are in }C^0(S_k)^2, 
	\left.\psi\right|_{(S_k)_+}=e^{-2i\pi\alpha_k}\left.\psi\right|_{(S_k)_-}\big\}.
\end{multline*}
The domain of the minimal operator is the closure in the graph norm of the set $\mathrm{Ans}_{\bA}^{(\min)}$.
We prove here that this domain coincides with the set $H_{\bA}^{1}(\S^3)^2$ 
of elements in $H^1(\Omega_{\uS})^2$ with the phase jumps across the $S_k$'s,
and that Lichnerowicz formula gives rise to the usual energy equality.

We first prove the energy equality. By density we can assume that the 
traces of $\psi\in \mathrm{Ans}_{\bA}^{(\min)}$ on both sides of $S_k$ are 
$C^1(S_k)$ for all $1\le k\le K$. Let $(\be_{1,k},\be_{2,k},\bN_{S_k})$ 
be an orthonormal basis of $\rT\,\S^3$ around $S_k$. 
By using Stokes' formula (as in Proposition~\ref{prop:gen_stokes}) we obtain the following equality:
\begin{multline*}
	\int |\cD_{\bA}^{(\min)}\psi|^2-\int |(\nabla\psi)\big|_{\Omega_{\uS}}|^2-\frac{3}{2}\int |\psi|^2\\
	=\sum_k\int_{S_k}i\vol_{S_k}\Big(\cip{\psi_{|_{(S_k)_+}}}{\sigma(\be_{1,k}^{\flat})\left.(\nabla_{\be_{2,k}}\psi)\right|_{(S_k)_+}
	-\sigma(\be_{2,k}^{\flat})\left.(\nabla_{\be_{1,k}}\psi)\right|_{(S_k)_+}}\\
	-\cip{\psi_{|_{(S_k)_-}}}{\sigma(\be_{1,k}^{\flat})\left.(\nabla_{\be_{2,k}}\psi)\right|_{(S_k)_-}
	-\sigma(\be_{2,k}^{\flat})\left.(\nabla_{\be_{1,k}}\psi)\right|_{(S_k)_-}}\Big),
\end{multline*}
where $\vol_{S_k}=\be_{1,k}^{\flat}\wedge\be_{2,k}^{\flat}$ is the volume form on $S_k$. 
The derivatives 
$$
	\left.(\nabla_{\be_{*,k}}\psi)\right|_{(S_k)_\pm}
$$
in the boundary terms are tangential to the Seifert surfaces $S_k$, 
hence they satisfy the same phase jump as the traces $\left.\psi\right|_{(S_k)_\pm}$
and the boundary terms cancel. 

We now pick an element $\psi\in H^{1}_{\bA}(\S^3)^2$ and show that
it lies in the minimal domain $\dom(\cD_{\bA}^{(\min)})$. 
After localizing with \eqref{def:part_unity}, we only have to prove that 
$\psi_k:=\chi_{\delta,\gamma_k}\psi$ is in the minimal domain
for any $1\le k\le K$. Up to multiplying by the phase jump functions 
$\overline{E}_{k}$ (see \eqref{def:curve_phasej_rem}), we can assume
that there is only one knot.
Furthermore, by a (local) gauge transformation $e^{i\phi_k}$ we can shift the phase jump
to $\{\theta_k=0\}$, where $\theta_k$ is the angle around the 
knot $\gamma_k$ (see~\eqref{sec:coord}).
All in all, we are left to study the function 
$$
	\wt{\psi}_k:=e^{i\phi_k}\overline{E}_{k}\psi_k,
$$
which, after a decomposition into its two parts $\cip{\xi_{\pm}}{\wt{\psi}_k}\xi_{\pm}$, can be assumed to be a scalar.

The function $e^{-i\alpha_k \theta_k}\wt{\psi}_k$ has no phase jump and thus
\[
	e^{-i\alpha_k \theta_k}\nabla \wt{\psi}_k
	=(\nabla+i\alpha_k\nabla \theta_k)\big( e^{-i\alpha_k \theta_k}\wt{\psi}_k\big).
\]
A direct computation shows that
$$
	\nabla\theta_k
	=-\frac{\tau_{S_k}}{h_k}\bT_k+\frac{1}{\sin(\rho_k)}\big(-\sin(\theta_k)\bS_k+\cos(\theta_k)\bN_k\big),
$$
(and $\nabla s_k=h_k^{-1}\bT_k$). Here $\tau_{S_k}$ denotes the relative torsion of $\gamma_k$ wrt $S_k$.
As we are only concerned with the behavior close to the knot $\gamma_k$,
it suffices to study the element as a function of $(s_k,\rho_k,\theta_k)$ and use the metric in $\T_{\ell_k}\times\R^2$
(see formulas~\eqref{eq:partiald_coord}).

We decompose the function $\cip{\xi_{\pm}}{e^{-i\alpha_k \theta_k}\wt{\psi}_k}\xi_{\pm}$ 
(as a function of $(s_k,\rho_k,\theta_k)$) relative to 
\[
	L^2(\T_{\ell_k}\times\R^2)=L^2(\T_{\ell_k})\otimes L^2(\R_+,\rho\d\rho)\otimes L^2(\R/(2\pi\Z),\d\theta).
\]
Thus, as $0<\alpha_k<1$, the condition 
$\norm{\nabla e^{-i\alpha_k \theta_k}\wt{\psi}_k}_{L^2(\S^3)^2}<+\infty$ implies
\[
	\int \frac{|\wt{\psi}_k|^2}{\rho_k^2}<+\infty.
\]
This implies that  $(\chi_{2^{-n}\delta,\gamma_k}\psi_k)_{n\in\N}$ converges to $\psi_k$ in the graph norm.
This sequence is a priori not in $\mathrm{Ans}_{\bA}^{(\min)}$, because the traces are only in $H^{1/2}(S_k)$.
However, by density we can approximate the $\chi_{2^{-n}\delta,\gamma_k}\psi_k$'s by elements in 
$\mathrm{Ans}_{\bA}^{(\min)}$ (in the graph norm, or equivalently, in the norm of $H^1(\Omega_{\uS})^2$) .

\subsection{Technical Lemmas}
\subsubsection*{Proof of Lemma~\ref{lem:model_c_def_sp}}
	We need to determine the deficiency spaces
	$$
		\Sigma_{\pm i} := \ker\big(\cD_{\T_{\ell},\alpha}^{(\max)} \mp i\big).
	$$

	We begin by determining $\Sigma_{+i}$, and thus want to find all solutions to
	$$
		\big(\cD_{\T_{\ell},\alpha}^{(\max)} - i\big)\psi=0
	$$
	in $\sD'(\T_{\ell}\times\R^2 \setminus \{u=0\})^2$.
	The translational symmetry of the problem allows us to Fourier transform in the $s$-direction, so that the operator
	takes the form
	$$
		\mathcal{J}_{\alpha} := \sigma_3 j+ \wt{\bsigma}(-i\wt{\nabla}_{u} + \bA_{u}),
	$$
	acting on wave functions $\phi \in \ell_2(\T_{\ell}^*) \otimes L^2(\R^2)^2$ with
	$$
		\phi(j,u)=(\sF_s\psi)(j,u) = \begin{pmatrix}\phi_+(j,u)\\ \phi_-(j,u) \end{pmatrix}.
	$$
	Acting with $(\mathcal{J}_{\alpha} + i)$ again leaves us with
	\begin{align}\label{eq:sa_line_1}
		(j^2 + (-i\wt{\nabla}_{u} + \bA_{u})^2 + 1)\phi_{\pm} = 0
	\end{align}
	for any $j \in \T_{\ell}^*$ fixed.
	After switching to polar coordinates in the $u$-plane and decomposing 
	$\phi_{\pm}(j,r,\theta) = \sum_{n \in \Z}\phi_{\pm,n}(j,r)e^{in\theta}$, \eqref{eq:sa_line_1} 
	gives us the set of equations (for fixed $j \in \T_{\ell}^*$, $n \in \Z$)
	$$
		\left(j^2 - \partial_{r}^2 - \frac{1}{r}\partial_r + \frac{(n+\alpha)^2}{r^2} +1\right)\phi_{\pm,n}(j,r) = 0.
	$$
	We rearrange the above to
	\begin{align*}
		(r^2\partial_r^2 + r\partial_r - (n+\alpha)^2 - (1+j^2)r^2)\phi_{\pm,n} = 0
	\end{align*}
	so we obtain the (scaled) modified Bessel equation. Any solution to the above equation can be written as
	\begin{align*}
		\phi_{\pm,n}(j,r) = c_{n}(j)K_{n+\alpha}(r\langle j\rangle) 
		+ d_{n}(j)I_{|n+\alpha|}(r\langle j\rangle)
	\end{align*}
	with $c_n(j),d_n(j) \in \C^2$. 
	Recall that if $0<\nu \notin \mathbb{N}$, then $K_{\nu}(s), I_{\nu}(s)$ have the 
	expansions (see \cite[Sections~9.6~\&~9.7]{AbrStegun})
	\begin{align*}
		K_{\nu}(s) &= \frac{\Gamma(\nu)}{2}\left(\frac{s}{2}\right)^{-\nu}\left(1 + \underset{s \to 0}{\mathcal{O}}(s^2)\right) 
		+ \frac{\Gamma(1-\nu)}{2\nu}\left(\frac{s}{2}\right)^{\nu}\left(1 + \underset{s \to 0}{\mathcal{O}}(s^2)\right),\\
		I_{\nu}(s) &= \frac{1}{\Gamma(\nu+1)}\left(\frac{s}{2}\right)^{\nu} + \underset{s \to 0}{\mathcal{O}}(s^{\nu+2}),
	\end{align*}
	and
	\begin{align*}
		K_{\nu}(s) = \sqrt{\frac{\pi}{2s}}e^{-s}\left(1+\underset{s \to \infty}{\mathcal{O}}(s^{-1})\right), \quad
		I_{\nu}(s) = \frac{e^{s}}{\sqrt{2\pi s}}\left(1+\underset{s \to \infty}{\mathcal{O}}(s^{-1})\right).
	\end{align*}
	One readily sees that if we require $\phi_{\pm,n} \in L^2(\R_+,rdr)$, then $d_n(j)=0$ $\forall n \in \Z$, 
	whereas $c_n(j)=0$ for $n\neq 0,-1$. 
	Writing 
	$$
		c_n(j)=(c_{n;+}(j),c_{n;-}(j))^{\text{T}}\in\mathbb{C}^2,
	$$
	we see that
	\begin{align*}
		\wt{\bsigma}(-i\wt{\nabla}+\bA_{\alpha})\psi \in L^2(\T_{\ell}\times \R^2)^2 
		\Rightarrow 
		\left\{ 
			\begin{array}{l}
			c_{0;+}(j)=c_{-1;-}(j)=0,\\
			c_{0;-},c_{-1;+}\in \ell_2(\T_{\ell}^*;\langle j\rangle^{-2}d\mu_{\T_{\ell}^*}(j)).
			\end{array}
		\right.
	\end{align*}
	Finally, from the condition $\cD_{\T_{\ell},\alpha}^{(\max)}\psi=i\psi$ one can deduce that
	$$
		c_{0;-}(j)=(1-ij)\langle j\rangle^{-1}c_{-1;+}(j).
	$$
	The above arguments can then be repeated for $(\mathcal{J}_{\alpha} + i)\phi=0$ to yield the condition
	$$
		c_{0;-}(j)=(-1-ij)\langle j\rangle^{-1}c_{-1;+}(j),
	$$
	which concludes the proof.
\qed

\subsubsection*{Proof of Lemma~\ref{lem:sa_knot_tech}}
We begin by writing
\begin{align*}
	h(s,\delta,\theta)\sin(\delta)
	&= \delta\left(1+(\cos(\delta)\tfrac{\sin(\delta)}{\delta}-1) 
	- \tfrac{\sin^2(\delta)}{\delta}(\cos(\theta)\kappa_g(s) + \sin(\theta)\kappa_n(s))\right)\\
	&= \delta(1+\delta w),
\end{align*}
where $w \in C^1(B_{\eps}[\gamma])$. Just as in the proof of Lemma~\ref{lem:reg_rho_dmax} one can
show that
$$
	f \in \dom\big(\cD_{\T_{\ell},\alpha}^{(\max)}\big), \supp f \subset \{\rho<const\}
	\Rightarrow \rho f \in \dom\big(\cD_{\T_{\ell},\alpha}^{(\min)}\big).
$$
Now set
$$
	f_{w} := \rho w f \chi_1 \in \dom\big(\cD_{\T_{\ell},\alpha}^{(\min)}\big),
$$
where $\chi_1$ is a cutoff function in $\rho$, so that
\begin{multline*}
	-i\lim_{\delta \to 0} \int_{0}^{\ell} \int_{0}^{2\pi} \delta
	\left[e^{i\theta}\overline{f_{w,-}}g_{+} + e^{-i\theta}\overline{f_{w,+}}g_{-}\right]\, \d\theta\d s\\
	= \cip{\cD_{\T_{\ell},\alpha}^{(\max)}f_w}{g}_{L^2} 
	- \cip{f_w}{\cD_{\T_{\ell},\alpha}^{(\max)}g}_{L^2} = 0.
\end{multline*}
\qed

\subsubsection*{Proof of Lemma~\ref{lem:control_Ngr}}
First observe that
$$
	\cD_{\T_{\ell},\alpha}^{(-)}f = (\wtd+\cE_1+\cE_0)f\big|_{\theta\neq 0} - (\cE_1+\cE_0)f\big|_{\theta\neq 0}.
$$
By \eqref{eq:D_max_close_knot} we immediately get that
$$
	\norm{(\wtd+\cE_1+\cE_0)f_{|_{\theta\neq 0}}}_{L^2(\T_{\ell}\times\R^2)^2}
	\leq C \norm{\cD_{\bA}^{(-)}(\chi_{\delta}\psi)}_{L^2(\S^3)^2},
$$
where the constant $C = \norm{\rho(h\sin(\rho))^{-1}}_{L^{\infty}(B_{\delta}[\gamma])}$ appears due to the discrepancy
of the volume forms.
Next, it is obvious that
$$
	\norm{\cE_0f}_{L^2(\T_{\ell}\times\R^2)^2} \leq 2C' \norm{\chi_{\delta}\psi}_{L^2(\S^3)^2},
$$
since for $\rho<\delta$ the matrix $\cE_0$ has entries bounded by $C'(\delta,\gamma)$.
For the term $(\cE_1f)_{|_{\theta\neq 0}}$ we use equations \eqref{eq:E1_comp_1} and \eqref{eq:E1_comp_2}
to obtain
$$
	\norm{(\cE_1f)_{|_{\theta\neq 0}}}_{L^2(\T_{\ell}\times\R^2)^2} 
	\leq C(\delta,\gamma) \left(\norm{\rho (\nabla_{\bT}(\chi_{\delta}\psi))_{|_{\theta\neq 0}}}_{L^2(\S^3)^2} + 
	\norm{\rho (\nabla_{\bG}(\chi_{\delta}\psi))_{|_{\theta\neq 0}}}_{L^2(\S^3)^2}\right).
$$
We then estimate
\begin{align*}
	\norm{\rho (\nabla_{\bT}(\chi_{\delta}\psi))_{|_{\theta\neq 0}}}_{L^2(\S^3)^2}
	&\leq \norm{(\nabla_{\bT}(\rho\chi_{\delta}\psi))_{|_{\theta\neq 0}}}_{L^2(\S^3)^2}\\
	&\leq \norm{\cD_{\bA}^{(\min)}(\rho\chi_{\delta}\psi)}_{L^2(\S^3)^2}\\
	&\leq \norm{\chi_{\delta}\psi}_{L^2(\S^3)^2} + \norm{\rho\cD_{\bA}^{(-)}(\chi_{\delta}\psi)}_{L^2(\S^3)^2}\\
	&\leq (1+\delta)\norm{\chi_{\delta}\psi}_{\bA}.
\end{align*}
A similar computation for $\norm{\rho (\nabla_{\bG}(\chi_{\delta}\psi))_{|_{\theta\neq 0}}}_{L^2(\S^3)^2}$ completes the proof.
\qed

\subsubsection*{Proof of Lemma~\ref{lem:char_sres_conv}}
	Assume (1). For $f \in \dom(\cD)$, we define $g \in \cH$ and $f_n\in\dom(\cD_n)$ by the relations: $(\cD-i)f=g=(\cD_n-i)f_n$.
	By strong resolvent continuity we have $\lim_{n\to+\infty}\norm{f-f_n}_{\cH}=0$, thus 
	$$
		\norm{\cD f-\cD_n f_n}_{\cH}=\norm{f-f_n}_{\cH}\underset{n\to+\infty}{\longrightarrow}0.
	$$

	Assume (2). Recall that $\ker(P)=\{(-\cD f,f): f \in \dom(\cD)\}$. For $f\in\dom(\cD)$, let $f_n\in\dom(\cD_n)$ such that
	$\lim_{n\to+\infty}(f_n,\cD_n f_n)=(f,\cD f)$. We write $(\wt{f}_n,\cD_n \wt{f}_n)=P_n(f,\cD f)$ 
	and $(-\cD_n f_n',f_n')=(1-P_n)(-\cD f,f)$. 
	We have
	\begin{align*}
		\norm{(f-\wt{f}_n,\cD f-\cD_n \wt{f}_n)}_{\cH^2}
		&=\norm{(P-P_n)(f,\cD x)}_{\cH^2}=\dist_{\cH^2}\big((f,\cD f),\cG_{\cD_n}\big)\\
		&\le \norm{(f-f_n,\cD f-\cD_n f_n)}_{\cH^2}\underset{n\to+\infty}{\longrightarrow}0.
	\end{align*}
	Similarly we get
	\begin{align*}
		\norm{\big[(1-P)-(1-P_n)\big](-\cD f,f)}_{\cH^2}
		\le \norm{(\cD_n f_n-\cD f,f-f_n)}_{\cH^2}\underset{n\to+\infty}{\longrightarrow}0.
	\end{align*}
	These two results imply (3), as we have:
	\begin{align*}
	P_n=P_n(P+(1-P))&=P_n P +(P_n-1)(1-P)+(1-P),\\
				      &=P_n P-(1-P_n)(1-P)+(1-P).
	\end{align*}

	Assume (3). For $g\in\cH$, we define $f\in\dom(\cD)$ and $f_n\in\dom(\cD_n)$ by the relations $(\cD-i)f=g=(\cD_n-i)f_n$.
	We also set $P_n(f,\cD f)=(\wt{f}_n,\cD_n\wt{f}_n)$. As $(f_n-\wt{f}_n)\in\dom(\cD_n)$, we have:
	\begin{align*}
		\norm{f_n-\wt{f}_n}_{\cH}^2+\norm{\cD_n(f_n-\wt{f}_n))}_{\cH}^2&=\norm{(\cD_n-i)(f_n-\wt{f}_n)}_{\cH}^2,\\
				&=\norm{(\cD-i)f-(\cD_n-i)\wt{f}_n}_{\cH}^2,\\
				&\le 2\big( \norm{f-\wt{f}_n}_{\cH}^2+\norm{\cD f-\cD_n \wt{f}_n}_{\cH}^2\big)
				\underset{n\to+\infty}{\longrightarrow}0.
	\end{align*}
	Thus 
	$$
		\norm{f-f_n}_{\cH}
		\le \norm{f-\wt{f}_n}_{\cH}+\norm{f_n-\wt{f}_n}_{\cH}\underset{n\to+\infty}{\longrightarrow}0.
	$$
\qed

\subsubsection*{Proof of Lemma~\ref{lem:decomp_lim}}
	Remark that for any $\bx\in\R^3$ we have 
	$$
		\big|\nabla^{\R^3}\chi_R(\bx) \big|
		\le \frac{C}{R} \norm{\nabla^{\R^3}\chi_1}_{L^{\infty}}\mathds{1}_{R/2\le |\bx|\le R}.
	$$
	We define the increasing family:
	\[
		Q_R:=\int |\chi_R|^2\left[|(\partial_{x_3}\psi_0)_{\R^3 \setminus \Sigma_{-\pi/2}}|^2
		+|(\sigma_{\perp}\cdot(-i\nabla_{\perp}^{\R^3}+\bA_{\perp})\psi_0)_{\R^3 \setminus \Sigma_{-\pi/2}}|^2\right]
	\]
	By Cauchy-Schwarz equality, the RHS in \eqref{eq:dirac_split_line_2} is equal to
	$$
		Q_R+\underset{R\to+\infty}{o}(\sqrt{Q_R}),
	$$
	while by \eqref{eq:estim_Dpsi_0} the LHS has a finite limit as $R$ tends to infinity.
	This shows that $\lim_{R\to+\infty} Q_R$ exists, giving \eqref{eq:decomp_lim}.
\qed

\subsubsection*{Proof of Lemma~\ref{lem:zero_m_s1}}
	Let us write $\nabla^{\Omega_3}$ the Levi-Civita connexion of $\S^3$ written in the chart of the stereographic projection.
	On $\mathcal{M}_2$, the spinor $\phi(\bx):=\psi(\iota(\bx))$ is a formal zero-mode for the Dirac operator
	$\wt{\cD}^{\Omega_3}_{\S^1}$ which acts away from $\D$ like 
	\begin{align*}
		\cD_\iota^{\Omega_3}
		&:=-i\sum_{k=1}^3\Omega_3(\iota(\bx))^{-1}\sigma\cdot\,((\iota_* X_k)^{\flat})\nabla^{\Omega_3}_{X_k}\\
		&=-i\sum_{k=1}^3\Omega_3(\iota(\bx))\sigma\cdot\,(g_{\R^3}(\iota_* X_k,\cdot))\nabla^{\Omega_3}_{X_k},
	\end{align*}
	where $(X_1,X_2,X_3)$ is any positive orthonormal basis of $\rT \mathcal{M}_2\simeq \mathcal{M}_2\times\R^3$.

	We have to change the basis of sections such that (away from $\D$) the Dirac operator has the usual expression
	$\cD^{\Omega_3}$, that is, we have to find $V_\iota(\bx)\in \SU(2)$ such that 
	$$
		V_\iota\cD_\iota^{\Omega_3}V_\iota^*=\cD^{\Omega_3}=-i\Omega_3^{-2}\sigma\cdot\nabla^{\R^3}\Omega_3.
	$$
	The corresponding zero mode will then be $\wt{\psi}(\bx)= V_\iota(\bx)\phi(\bx)$.
	It is clear that the associated 
	rotation $R_\iota(\bx)$ must send any orthonormal basis 
	$(X_k(\bx))_{1\le k\le 3}$ onto $(\tfrac{\d\iota(\bx)}{|\d\iota(\bx)|}X_k(\bx))_{1\le k\le 3}$.
	Let us prove that one possibility is given by $V_\iota(\bx):=U_{\iota(\bx)}(-i\sigma_3)U_{-\bx}$.
	We write $O(\bx)\in \SO(3)$ the associated rotation of $U(\bx)$, defined by:
	$U_{\bx} (\sigma\cdot v) U_{\bx}^{*}=\sigma\cdot\big(O(\bx)v \big),\ v\in\R^3$.
	Consider the global orthonormal basis of $\rT\S^3$: 
	$$
		u_1:=(i\overline{z}_1,i\overline{z}_0), \quad u_2:=(\overline{z}_1,-\overline{z}_0), \quad u_3:=(iz_0,iz_1).
	$$
	The matrix $O(\st(\bp))$ is the change of basis from the canonical basis
	of $\R^3$ to 
	$$
		\Big(\Omega_3(\st(\bp))\st_*(u_j)(\st(\bp))\Big)_{1\le j\le 3}.
	$$
	The rotation $\mathrm{sw}:(z_0,z_1)\mapsto (z_1,z_0)$ satisfies 
	$$
		\mathrm{sw}_*(u_3)=u_3, \quad \mathrm{sw}_*(u_j)=-u_j \textrm{ for } j\in\{ 1,2\}.
	$$
	Thus the rotation $R_\iota(\bx):=|\d\iota(\bx)|^{-1}\d\iota(\bx)$ is associated 
	to the $\SU(2)$-matrix $U_{\bx}i\sigma_3 U_{-\iota(\bx)}$.
\qed





\end{document}